\documentclass[reprint,prl,
superscriptaddress,
 amsmath,amssymb,
 aps,
]{revtex4-2}
\usepackage{graphicx}
\usepackage{titlesec}
\titlespacing{\subsubsection}{2pt}{2.5ex}{2.5ex}
\titlespacing{\subsection}{2pt}{3.5ex}{3.5ex}
\usepackage{dcolumn}
\usepackage{bm}
\usepackage[T1]{fontenc} 
\usepackage[english]{babel} 
\usepackage[utf8]{inputenc}
\usepackage[table,usenames,dvipsnames]{xcolor}
\usepackage{array}
\usepackage{booktabs}
\usepackage{multirow}
\usepackage[psdextra]{hyperref}
\usepackage{bookmark}
\usepackage{amsfonts}
\usepackage{amsmath,amssymb,amsfonts,mathtools,amsthm}
\usepackage{graphicx}
\usepackage{bbm}
\usepackage{tikz}
\usetikzlibrary{quantikz2, calc, positioning}
\usetikzlibrary{arrows.meta, positioning}
\usetikzlibrary{fit, shapes.geometric, backgrounds}
\usepackage{enumitem}
\usepackage[normalem]{ulem}
\newcounter{mycountera}
\newcounter{mycounterb}
\newcounter{mycounterc}
\newcounter{mycounterd}
\newcounter{mycountere}

\renewcommand{\themycounterb}{\alph{mycounterb}}

\renewcommand{\themycountere}{\alph{mycountere}}
\makeatletter
\renewcommand{\p@mycountera}{\thefigure}
\renewcommand{\p@mycounterb}{\thefigure}
\renewcommand{\p@mycounterc}{\thefigure}
\renewcommand{\p@mycounterd}{\thefigure}
\renewcommand{\p@mycountere}{\thefigure}
\makeatletter

\usepackage{hyperref}
\usepackage[most]{tcolorbox}

\newtcbtheorem[auto counter]{algo}%
  {Protocol}{fonttitle=\bfseries\upshape, fontupper=\slshape,
     arc=3mm, colback=white,colframe=vinho,center title}{algo}
\definecolor{vinho}{HTML}{A50021}
\definecolor{vinho}{HTML}{ea801c}
\definecolor{black}{HTML}{000000}
\newtheorem{theorem}{Theorem}
\newtheorem{proposition}{Proposition}
\newtheorem{lemma}{Lemma}
\newtheorem{definition}{Definition}

\begin{document}

\preprint{APS/123-QED}

\title{Quickly extracting fidelity decay rates in random circuit benchmarking experiments}%

\author{Ake Köhne}%
 \email{Corresponding author: ake.koehne@cwi.nl }
\affiliation{CWI \& QuSoft, Amsterdam}%
\author{Stavros Efthymiou}%
\affiliation{Quantum Research Center, Technology Innovation Institute, Abu Dhabi}%
\author{Joanna Majsak}%
\affiliation{Quantum Research Center, Technology Innovation Institute, Abu Dhabi}%
\affiliation{Institute for Theoretical Physics, University of Cologne, Cologne, Germany}
\author{Jadwiga Wilkens}%
\affiliation{Johannes Kepler Universität Linz, Austria}%
\author{Bram Evert}%
\affiliation{Rigetti Computing, 775 Heinz Avenue, Berkeley, California 94710, USA}%
\author{Andrew D. Patterson}%
\affiliation{Rigetti Computing, 775 Heinz Avenue, Berkeley, California 94710, USA}%
\author{Jonas Helsen}%
\affiliation{CWI \& QuSoft, Amsterdam}%
\author{Markus Heinrich}%
\affiliation{Institute for Theoretical Physics, University of Cologne, Cologne, Germany}
\author{Ingo Roth}%
\affiliation{Quantum Research Center, Technology Innovation Institute, Abu Dhabi}%

\date{\today}

\begin{abstract}
Randomized benchmarking (RB) with structured circuits is experimentally scalable, but gives biased short-depth fidelity estimates because incomplete ensemble mixing introduces transient modes. 
Waiting for these transients to decay in deeper circuits can push the signal below the noise floor. This is a well-established obstacle, e.g.\ in linear cross-entropy benchmarking (LinXEB). 
Here we overcome this obstacle solely through classical post-processing.  We introduce a one-parameter family of estimators that interpolate between the standard LinXEB-type post-processing and a new post-processing technique we call the trace filter function.  
Under practical assumptions, the trace filter function provably suppresses transient modes, exposing an interpretable fidelity decay at short depth, with theoretical guarantees comparable to standard RB, at the cost of an increased variance. 
We bound the variance, analyze the evaluation complexity of the new estimators, and develop practical post-processing algorithms.  
Experiments on a Rigetti superconducting processor validate the method for one-dimensional circuits of up to 12 qubits: where the standard estimator requires depth beyond a 40-layer experimental window, the trace filter function yields consistent estimates below 10 layers.  
Trace-filtered randomized benchmarking therefore extends reliable fidelity characterization to relevant structured circuit ensembles and short-depth regimes that are inaccessible to conventional estimators.
\end{abstract}

\maketitle

In the quest to develop large-scale and robust quantum computers, an important role is reserved for protocols aiming to characterize the performance of these devices 
informing 
fabrication, calibration, noise mitigation, error correction and algorithmic resource estimation \cite{eisert2020quantum,kliesch2021theory,hashim2025practical,proctor2025benchmarking}.
 
Many of these characterization methods can be grouped under the class of randomized benchmarking (RB) protocols \cite{helsen2022general,RB1,RB2,RB3,erhard2019characterizing,RB4,RB5,RB6,hines2024fully}, which aim at characterizing the average error rate of a gate set.
This is done by estimating the decay rate of the implementation fidelity of randomly drawn quantum circuits with increasing depth. 
Sample-efficiency and robustness to state preparation and measurement (SPAM) errors make RB protocols particularly attractive and widely used in practice. 

In practice, standard RB with random Clifford unitaries is limited to very few qubits, typically one or two, because of the use of deep circuits.
More scalable benchmarks make use of random quantum circuits \cite{francca2018approximate,brown2018restricted,RB3,RB4,heinrich2022randomized,hines2024fully}.
Most prominently, this includes the linear cross-entropy benchmarking (LinXEB) protocol \cite{boixo2018characterizing,RB4}, originally designed for certifying quantum advantages in sampling tasks, if one considers the decay of the circuit-averaged LinXEB fidelity with circuit depth \cite{helsen2022general,heinrich2022randomized}. 

The assumptions under which RB protocols yield reliable and interpretable benchmarks are well understood \cite{helsen2022general,heinrich2022randomized,ChenDingHuang2022RandomizedBenchmarkingBeyondGroups,figueroaRomero2021nonMarkovian,brillant2025nonMarkovian,dalzell2024random}. For Markovian, time-stationary, and gate-dependent noise, RB signals are expected to be dominated by a single exponential decay whose decay parameter is related to the fidelity of an effective average noise process. This simplification relies on averaging over sequences that explore the relevant state space sufficiently isotropically, so that only isotropic components of the noise remain visible. 
For random-circuit RB, or more generally for non-uniform gate distributions, the single-decay regime is reached after a minimal sequence length, at which the circuit distribution has mixed sufficiently toward the corresponding uniform ensemble.

Before the circuit distribution has mixed, the RB signal contains additional transient components associated with the convergence of the sampled circuit ensemble. Neglecting these components biases single-exponential estimates of the decay rate, an issue already recognized by comparisons of the NIST and uniformly sampled Clifford protocols \cite{boone2019gatesets}, and prominently noted in LinXEB \cite{rinott2022statistical,chen2023linear,liu2021benchmarking,ware2023sharp}.
One can avoid this bias by fitting only sufficiently deep, well-mixed sequences; however, this is often experimentally infeasible because noise may degrade the signal before the mixing regime is reached.

In LinXEB, finite-depth mixing effects are typically mitigated either by engineering circuit topologies and local gate ensembles for fast mixing \cite{boixo2018characterizing,RB4,mi2021scrambling}, or by introducing ad hoc ratio estimators that normalize the measured signal by the corresponding noiseless signal of the same non-uniform circuit ensemble \cite{rinott2022statistical,mark2023ergodic}.
The latter approach can remove bias arising from the ideal circuit distribution, but its status as an RB estimator is less clear: it is not evident that the ratio signal remains governed by a single interpretable decay parameter related to an average fidelity, under realistic gate-dependent noise. 
At the same time, slowly mixing circuit architectures arise naturally in settings where a restricted topology \cite{harrowMehraban2023designs} or entangling pattern is precisely the object one wants to benchmark. A particularly relevant example is one-dimensional brickwork circuits, which appear, for instance, as digital models of local scrambling and many-body dynamics \cite{mi2021scrambling,hartnett2026fast}.

In this work, we address the finite-depth mixing bias by redesigning the estimators rather than changing the circuit ensemble or normalizing by an ideal reference signal. We introduce a one-parameter family of \emph{filter functions} that interpolates between the standard LinXEB-type estimator and a new estimator, called the \emph{trace filter function}. See Figure \ref{fig:fig1_explanation} for an overview of the resulting protocol.

\begin{figure*}
    \centering
    \includegraphics[width=\linewidth]{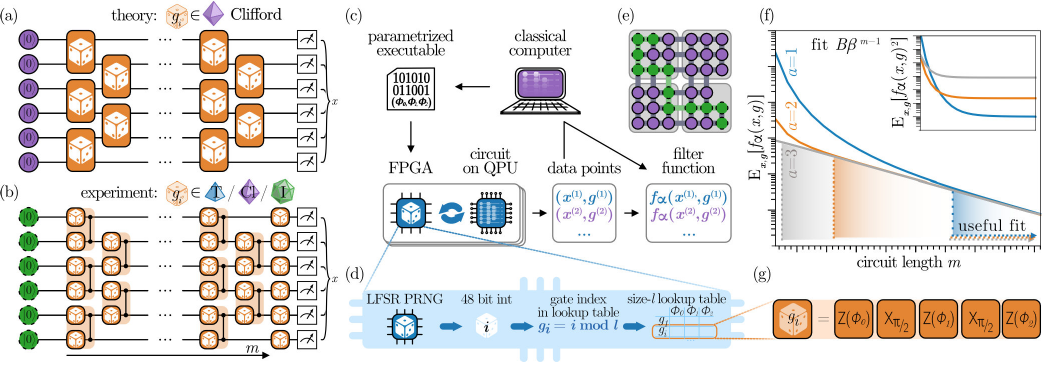}
    \caption{ Overview of setup.
    (a) Circuits used in simulations. Each random circuit instance $g$ is constructed from a brickwork pattern of $m$ layers of 2-qubit Clifford gates.
    (b) Circuits used in experiments. Each 2-qubit gate from (a) is now a layer of 1--qubit gates, each sampled from the  ensemble (tetrahedral, Clifford or icosahedral), followed by a CZ gate. The last layer is followed by a layer of 1--qubit random gates. 
    (c) Steps of the experiment exploiting on-chip randomness.
    The classical computer prepares a program where each random gate is parametrized by three angles $\phi_0, \phi_1, \phi_2$.
    The parametrized executable is uploaded to the FPGA, where the angles are randomly sampled.
    The FPGA sends program instances to the quantum computer (QPU).
    Each instance is run and measured in single-shot mode.
    {Real-time classical control on FPGA represented by thick blue arrows.}
    A data point saved for each random circuit instance includes: the executed random circuit $g^{(i)}$ and the measurement outcome $x^{(i)}$. 
    These data points, from single-shot random circuits per sequence length, are processed on the classical computer.
    The value of the filter function is calculated for each.
    (d)  Randomness on the FPGA.
    A 48-bit linear-feedback-shift-register (LFSR) pseudo-random number generator (PRNG) is used to generate a random integer $i$.
    It is modulo-divided by the size of the lookup table of the target group, which determines the random gate. 
    (e) The layout of the 36-qubit superconducting chip on which the experiments were executed. {The QPU consists of 4 9-qubit chiplets.}
    {One qubit was not operational and is not shown.}
    The 12 (6/10) qubits in linear layout were chosen as a snaking pattern in the device (green circles, dashed border corresponds to the 12-qubit experiment).   
    (f) The theoretical scaling of the first (inset: second) moment of the filter function, for different values of $\alpha$, in lin--log scale.
    A decay parameter $\beta$ is fitted to the data for the first moment. The shaded areas indicate the regimes of $m$ from which a useful fit can be obtained, for different values of $\alpha$.
    (g) Each 1--qubit random gate is compiled as $Z_0X_{\pi/2}Z_1X_{\pi/2}Z_2$, with $X_{\pi/2}$ physical rotation by $\pi/2$, $Z_i$ virtual rotation by $\phi_i$.}
\label{fig:fig1_explanation}
\end{figure*}%
For this trace filter function we rigorously argue that the associated RB experiment has, for time-stationary Markovian gate-dependent noise and practical local-symmetry assumptions, a mixing length of zero, thus enabling the extraction of fidelities from much shorter circuits. Estimation in this short circuit-regime comes at a cost of increased variance, and an interpolating family of functions allows us, by tuning a parameter $\alpha$, to trade off variance for mixing length. We provide analytic bounds on the long-sequence length variance in the noiseless regime for different values of $\alpha$ and corroborate these with numerical evaluations and simulations, establishing the optimality of intermediate filter functions in different parameter regimes. 

We further characterize the classical cost of evaluating these filter functions. Although worst-case evaluation is hard even for Clifford circuits, we give efficient algorithms in the short- and long-depth regimes and a practically scalable stabilizer-based implementation for Clifford circuits. 

We implement our procedure experimentally for different random 1D brickwork circuit ensembles on up to 12 qubits of a Rigetti superconducting device which was optimized for high-throughput circuit sampling. 
We find in the experiments (Fig.~\ref{fig:fig3_rigetti_exp}) that the trace filter function allows one to extract fidelities at sequence lengths substantially below the mixing length, enabling faster and more accurate fidelity characterization of large devices. For instance, at 12 qubits we find that the standard estimator
requires depth beyond a $40$-layer experimental window while the trace filter function yields consistent estimates
below $10$ layers. This implies trace-filtered randomized benchmarking can be used for reliable fidelity characterization in short-depth regimes that are inaccessible to conventional estimators.

    \section{Trace-filtered randomized benchmarking}

A large class of RB protocols can be formulated by defining a probability measure $\nu$ on the group of $n$--qubit unitaries $\operatorname{SU}(2^n)$ and a real-valued function $f$ for the post-processing, called the \emph{filter function} \cite{helsen2022general,heinrich2022randomized}. 
The \emph{experiment} then consists of repeating the following primitives: (i) prepare an initial state $\ket 0$,  (ii) apply a random sequence $g_1, \ldots, g_m$ of $m$ unitaries $g_i$, drawn i.i.d.\ from $\nu$, (iii) measure in the computational basis and record the resulting bit string $x$ and the overall unitary $g = g_m \cdots g_1$.  The primitive is repeated $N$ times for different sequence lengths $m$ yielding multiple tuples $(x^{(l)}, g^{(l)})$ of bit strings and unitaries.  This can be seen as a variant of the standard RB protocol~\cite{RB2} where the global inversion gate is moved to the classical post processing. 

In the \emph{classical post-processing}, we calculate the RB signal as the empirical average of the filter function $f$ for each sequence length $m$
\begin{equation}
    \hat F(m) = \frac1N \sum_{l=1}^N f(x^{(l)}, g^{(l)})\,. 
\end{equation}
Subsequently, fitting an exponential in $m$ to the RB signal allows one to extract a decay parameter $\beta$---the result of the RB protocol.  
For a perfect gate implementation, $\beta$ is $1$ even in the presence of state preparation and measurement errors.  
This means $\epsilon = -\log(\beta) \approx 1 - \beta$ quantifies the infidelity of the gate implementation.  
In certain settings, one can  establish a rigorous interpretation of $\epsilon$ as being proportional to the average gate-infidelity, see e.g.~Ref.~\cite{helsen2022general}. 

The most straightforward choice of measure $\nu$ is the uniform measure over the entire Clifford group. 
However, long sequences of compiled multi--qubit Clifford unitaries lead to circuits with challenging depth on current devices, limiting standard RB to very few qubits.  
For circuit-level benchmarks, one therefore often resorts to defining $\nu$ as a distribution over quantum circuits of bounded depth. For example, we in this paper i.a.\ choose $\nu$ to have uniform support on one-dimensional `brickwork layers' with fixed two-local entangling gates and all combinations of one-local gates from a discrete gate set, see Figure~\ref{fig:fig1_explanation}. 

The most prominent example of a filtered RB protocol with random quantum circuits arises when one performs linear cross-entropy benchmarking (LinXEB) and reports the decay of the LinXEB score with the circuit depth, averaged over multiple circuits.  
Here the filter function is given by 
\begin{equation}
    f_{1}(x,g)= (2^n+1)\left(|\langle x| g |0\rangle|^2 - 2^{-n} \right)\,,\label{eq:linXEB}
\end{equation}
where we used a $2^n+1$ normalization factor compared to $2^n$ in the original proposal. 

In order to regard the decay parameter $\epsilon$ as an infidelity, one requires a robust interpretation, under clear assumptions, justifying the single-exponential model. 
If $\nu$ is the uniform measure, 
the RB signal with the standard filter function $f_1$ is in expectation well-approximated by 
\begin{equation}\label{eq:single-decay}
    \mathbb E[\hat F(m)] = \mathbb E[f_1(m)]\approx 
    B \beta^{m-1}\,, 
\end{equation}
with $B$ a constant, under general non-local, gate-dependent, time-stationary, Markovian noise~\cite[Thm.~16]{helsen2022general}.  
The approximation error in the above equation is due to the gate-dependence of the noise and is exponentially suppressed in the sequence length $m$.
The infidelity in this case formally coincides with the depolarizing strength of an effective gate-independent noise description.  

A perfectly implemented random quantum circuit converges, for long sequences, towards the uniform measure on the group it generates. 
For sufficiently small noise strength, these convergence properties are unaltered.
As a consequence, the RB signal when using a non-uniform measure like a random quantum circuit 
has additional signal components~\cite{heinrich2022randomized}: 
\begin{equation}\label{eq:subdominant-decays}
    \mathbb E[f_1(m)]  \approx B \beta^{m-1} + \sum_i C_i \gamma_i^{m-1}\,,
\end{equation}
that decay in $m$ faster than the dominant component $|\gamma_i| < \beta$, which is typically identified as a fidelity measure of an average gate with respect to the distribution $\nu$.
Thus, there exists a minimal sequence length $m_0$ such that the single decay approximation still holds for $m>m_0$. 
In principle, for sufficiently long sequences in a non-uniform RB protocol, we can thus still estimate the interpretable decay parameter $\beta$. 
Ref.~\cite{heinrich2022randomized} gives rigorous bounds on the signal approximation and minimal sequence length for non-uniform filtered RB under gate-dependent noise, see also Ref.~\cite{ChenDingHuang2022RandomizedBenchmarkingBeyondGroups}.  

However, in practice, long sequence lengths are associated to small signal, requiring comparatively many samples to estimate the decay rate.
Especially for quantum circuits that do not `mix' sufficiently fast, obtaining an unbiased infidelity estimator using the standard LinXEB protocol becomes infeasible as the signal has decayed before the subdominant modes become negligible. 
A practically relevant and particularly challenging example here comes from one-dimensional random circuits, e.g.\ for benchmarking layers of entangling gates with particular geometric restrictions, which have empirically observed large mixing lengths. 

We here overcome this problem of large mixing lengths by introducing novel filter functions \footnote{The original proposal of the trace filter function was sketched in the unpublished manuscript \cite{heinrich2022randomized} by two of the current authors.} that synthetically reduce the mixing length.  At the heart of our construction is the following structural insight: 
Formally, the RB signal with standard filter function is of the form
\begin{equation}\label{eq:RBsignalT}
    \mathbb E[f_1(m)] = \operatorname{Tr}\big[\rho\, T^m(\tilde M)(\tilde \rho)\big],  
\end{equation}
where $T$ is an operator on quantum channels depending on the gate-implementation, $\tilde M$ is the (noisy) read-out channel and $\tilde \rho, \rho$ are the initial state with and without noise, respectively.  
(The operator $T$ can be understood as a noisy channel twirl or a Fourier transform of an implementation map that relates $g$ to its noisy implementation as a quantum channel \cite{merkel2021randomized,heinrich2022randomized,helsen2022general}.) 
The spectrum of $T$ is the set of decay rates of the RB signal.  
For uniform $\nu$, $T$ is, to a good approximation, of unit rank, leading to the single exponential decay, Eq.~\eqref{eq:single-decay}.  
For a random quantum circuit the spectrum of $T$ has many additional components and multiple eigenspaces overlap with $\tilde M$ and $\rho \otimes \tilde \rho$.  

Importantly, the factor $\tilde \rho$ arises from the Born probabilities describing the experiment and cannot be changed in the post-processing.  
Our key observation is that, by making careful use of the symmetries of the random quantum circuit, we can overcome this 
apparent restriction and realize an RB signal in the post-processing that takes the form of a trace-inner product in operator space with some operator $P$:
\begin{equation}\label{eq:RBsignalTrT}
    \mathbb E[\hat F(m)] = \operatorname{Tr}[ P\, T^m(\tilde M) ]\,.  
\end{equation}
Then we choose $P$ to have high overlap with the dominant unit-rank component of $T$ effectively suppressing the subdominant decays, i.e.\ the coefficients $C_i$ in Eq.~\eqref{eq:subdominant-decays}. 

Assuming local Clifford symmetry, we can realize this trace inner product by suitably reweighting symmetry-related measurement outcomes in post-processing. 
To this end, we define the family of filter functions 
\begin{equation}
    f_{\alpha}(x,g)=(2^n+1)\left[\langle x| g\, \xi(\alpha)\, g^\dagger|x\rangle - (1 + \alpha)^{-n}\right]\label{eq:family}\,, 
\end{equation}
with $\alpha \in [1, 3]$ and, using Pauli operators $I$ and $Z$, 
\begin{equation}
    \xi(\alpha) = \frac1 {(1+\alpha)^n} \left (I + \alpha Z \right)^{\otimes n}\,. 
\end{equation}
For $\alpha = 1$, the filter function coincides with the standard filter function $f_1$ of Eq.~\eqref{eq:linXEB}.  
Using the function $f_3$, which we call the trace filter function, the RB signal turns out to take the form of Eq.~\eqref{eq:RBsignalTrT}. 

More precisely, we establish the following result. 
\begin{proposition}[De-biasing\ (informal)]\label{prop: bias} 
If $\nu$ is right-invariant under local Clifford rotations and the experiment admits a time-stationary, Markovian, gate-dependent  description that does not break the local Clifford-symmetry of the $\nu$-channel twirls $T$, 
then 
\begin{equation}
 \mathbb E[f_3(x,g)] \propto \operatorname{Tr}[P T^m(\tilde M)]\,, 
\end{equation}
where $\tilde M$ is the read-out channel under noise and $P$ is the fixed point of a self-adjoint $T$ in the absence of noise. 
\end{proposition}

The formal statement and proof are given in Appendix~\ref{appendix:signal}.

The expected $f_1$ signal exhibits a mixing phase dependent on sequence length $m$ even in the absence of noise.
The above stated proposition implies that when the gate-implementation is perfect the expected RB signal for the trace filter function $f_3$ becomes proportional to $\operatorname{Tr}[P\tilde M]$, independent of $m$. 
Hence, for the trace filter function, the subdominant decays are perfectly suppressed.  
Our method relies on the expectation that this suppression 
is robust under noise.  
The underlying assumption that 
the noise does not break the local Clifford symmetry of $\nu$-twirls is rather weak and often justifiable in practice.  
For example, the assumption holds if imperfections in  
single-qubit rotations can be neglected compared to the implementation of entangling layers. Hence, performing trace filtered RB leads to an unbiased estimate for the decay rate even in the limit of short, i.e.\ shorter than the mixing length, sequence lengths. 

Although the subleading signal components introduced by mixing are suppressed by using $f_3$, this generally comes at the price of a decrease in signal-to-noise ratio in the individual signal points. 
To demonstrate this, we analytically calculate the variance of $f_\alpha$ when $(x, g)$ are obtained in a noiseless experiment using uniformly random global Clifford unitaries, corresponding to the long-sequence length limit in the brickwork experiment, in Appendix~\ref{appendix:varnonoise}. 
\begin{proposition}\label{prop:var_no_noise}
Let $g$ be drawn uniformly at random from the $n$--qubit Clifford group and $x$ drawn from $p(x|g) = |\langle  x | g | 0 \rangle|^2$. Then 
\begin{equation}\label{eq:variance}
    \operatorname{Var}[f_\alpha(x,g)] \in \mathcal{O}\left(
    \left[1 +  
    \left(\frac{\alpha - 1}{\alpha +1}\right)^2\right]^n\right)\,.
\end{equation}
\end{proposition}
Note that the variance is quadratic in $\alpha-1$ near the standard filter function, but exponential in $n$ for every fixed $\alpha > 1$. Over $\alpha \in [1,3]$, the exponential base ranges from $1$ to $5/4$. 
This observation motivates introducing the family of filter functions parametrized by $\alpha$.  The parameter $\alpha$ allows one to trade off variance for bias in the signal. 
Picking $\alpha$ closer to $\alpha=3$ decreases the mixing length of the RB signal at the price of increasing the variance. 

We can also model the variance in the short sequence length limit, by analytically calculating the variance for a single layer of uniformly random local Cliffords. We performed this calculation for $\alpha=1$ and $\alpha=3$ in Appendix~\ref{appendix:varnonoise}. Curiously the variance inverts in this limit, being largest for $\alpha=1$. However this is an artifact of the fact that for $\alpha<3$ the signal at this sequence length is exponentially large. Computing the signal-to-noise ratio restores the expected monotonicity in $\alpha$.

As we exemplify in numerical simulations below, given the parameters of the experiment and noise levels, this bias-variance trade-off leads to an  optimal choice of $\alpha$ that needn't be $\alpha=1$ or $\alpha=3$.

A further trade-off of the filter function is a potential increase in the complexity of the classical post-processing.  
In LinXEB experiments typical gate sets generate classically intractable sampling problems, accordingly also the evaluation of the filter function $f_1$ is not efficient and requires sizable computational resources in practice.  
In a noise benchmarking setting also experiments with classically efficiently simulable circuits, e.g. Clifford circuits, are well-motivated.  
When using the Clifford group the standard filter function $f_1(x,g)$ is computable exactly in polynomial time in $n$ using the stabilizer phase-space formalism. 

For $\alpha > 1$ we show in Appendix~\ref{appendix:filters} that even for Clifford circuits evaluating the filter function takes exponential time in $n$ in the worst case, as the problem contains the weight-enumerator problem for binary linear codes, implying $\#P$-hardness. 
However, this worst-case hardness is mitigated in important regimes. For uniformly random global Clifford unitaries, we find that direct evaluation in expectation takes time $O(1)$, making the protocol efficient in the long sequence length limit. 

In the short sequence length limit we find that the weight-enumerator problem typically involves exponentially many terms in expectation for random single-qubit Clifford rotations.  
In this setting however, the filter function factorizes and can be evaluated in $O(n)$ time regardless of the number of terms.  More generally, for geometrically local one-dimensional circuits of $O(\log(n))$ depth a tensor network algorithm is efficient. 
These two pieces of evidence suggest that a formally efficient algorithm might exist for short sequence lengths.
We find that the $O(n^2)$ memory complexity of evaluating the trace filter function and the straightforward parallelization of the evaluations of different shots allow for a sufficiently scalable implementation for Clifford unitaries seen in practice using the stabilizer formulation of Appendix~\ref{appendix:filters}, even without resorting to tensor networks. We implemented the algorithm in python using the Stim and Galois packages \cite{gidney2021stim,Hostetter_Galois_2020}. For non-Clifford unitaries, we performed statevector simulations using qsim \cite{qsim}.

\section{Numerical evaluation of moments}
\begin{figure*}[t]
    \centering
    \setcounter{mycountere}{0}
    \begin{minipage}[b]{4.45cm}
        \centering
        \includegraphics[width=\linewidth]{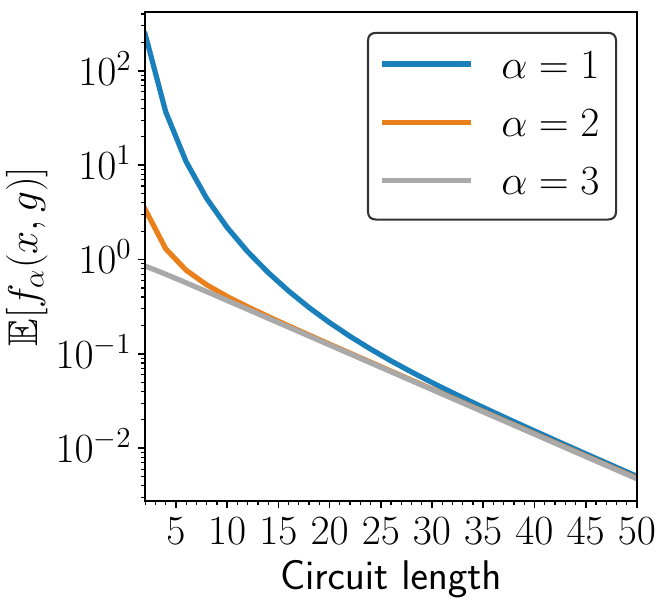}
        \refstepcounter{mycountere}\label{fig:figure2-signal}
        \footnotesize(\themycountere)
    \end{minipage}\hfill
    \begin{minipage}[b]{4.45cm}
        \centering
        \includegraphics[width=\linewidth]{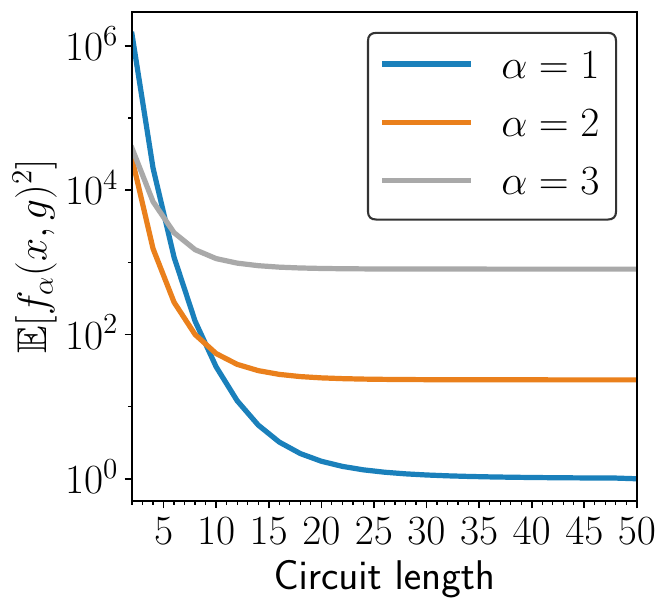}
        \refstepcounter{mycountere}\label{fig:figure2-second-moment}
        \footnotesize(\themycountere)
    \end{minipage}\hfill
    \begin{minipage}[b]{4.45cm}
        \centering
        \includegraphics[width=\linewidth]{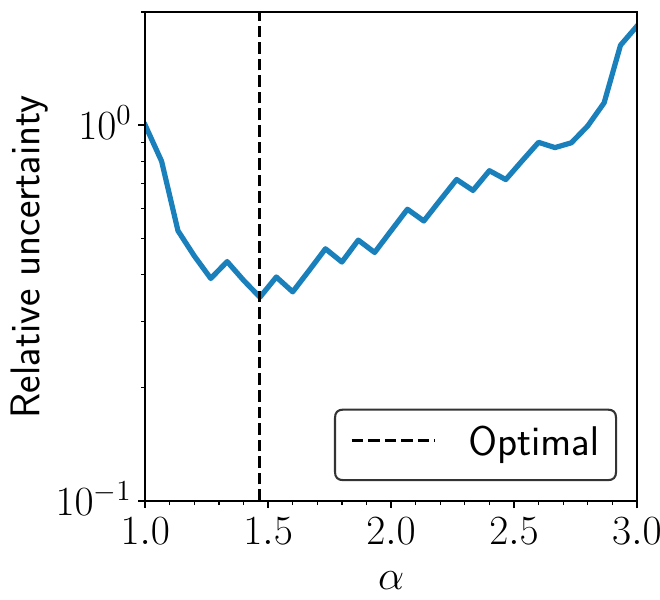}
        \refstepcounter{mycountere}\label{fig:figure2-uncertainty}
        \footnotesize(\themycountere)
    \end{minipage}\hfill
    \begin{minipage}[b]{4.45cm}
        \centering
        \includegraphics[width=\linewidth]{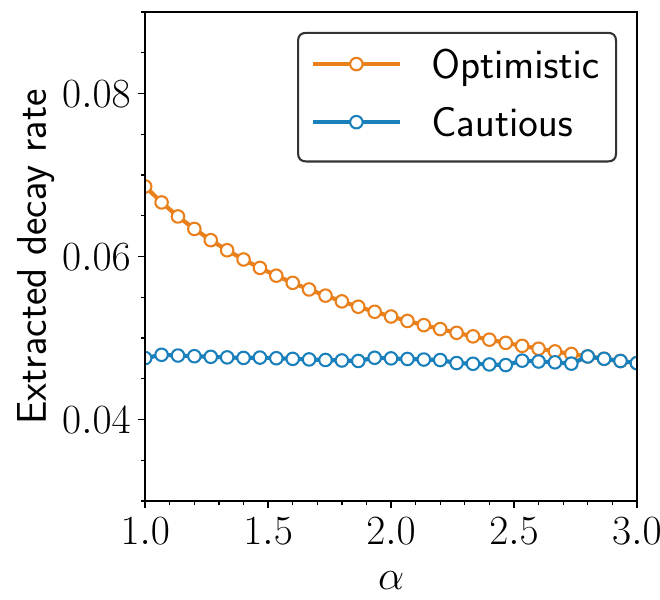}
        \refstepcounter{mycountere}\label{fig:figure2-decays}
        \footnotesize(\themycountere)
    \end{minipage}
    \caption{
    Numerical evaluation of filtered RB for an open 1D brickwork circuit on 30 qubits, with local two-qubit Clifford gates interspersed with identical single-qubit depolarizing noise. The moments are evaluated exactly using local applications of Weingarten calculus and tensor-network contractions, as detailed in Appendix~\ref{appendix:signal_sim}. (a) Filtered RB signals, showing the reduction of the transient regime with increasing $\alpha$ and the nearly single-exponential decay of the trace filter $\alpha=3$. (b) Corresponding second moments. (c) Uncertainty of the decay-parameter estimate, normalized by the uncertainty for $\alpha=1$, showing an optimal value of around $\alpha=1.5$ for this setting. (d) Decay values extracted using the cautious and optimistic fitting strategies described in the main text, using data from the 6-qubit IBM Guadalupe emulator (Qiskit).
    }
    \label{fig:figure2}
\end{figure*}

In the previous section, we stated analytic expressions of the variance for a single Clifford layer and global random Clifford unitaries, the limits of short and deep circuits, respectively. 
Deriving analytical expressions for the intermediate regime is more involved.  
We therefore resort to numerical evaluations of the formulas.  
As explained in Appendix~\ref{appendix:signal_sim}, the moments of $f_\alpha$ for random circuits can be explicitly evaluated using local applications of Weingarten calculus and tensor network contractions.  

We will consider an open one-dimensional brickwork circuit on $30$ qubits (see Figure~\ref{fig:fig1_explanation}a), where each brick is uniformly sampled from the $2$-qubit Clifford group.  
We use a gate-independent noise channel acting  after every $2$--qubit gate. 
The noise channel is simply the product of single--qubit depolarizing noise with depolarization probability $p=5\cdot 10^{-3}$ of identical strength for all qubits and layers.
Figures \ref{fig:figure2}a and \ref{fig:figure2}b show the first moment (bias) and second moment (variance) for the filter functions for different values of $\alpha$. 

We observe that the RB signal with $\alpha<3$ deviates significantly from a single-exponential decay for circuit lengths of up to $30$ layers.  
With increasing $\alpha$ the deviation is reduced, in accordance with our expectation. 
The trace filter function $\alpha=3$ has a mixing length indistinguishable from zero, numerically validating the analytical statements in the previous section. 
The variances decay to asymptotically constant values for deeper circuits for all $\alpha$. Worthy of note here is that in the short sequence length regime the variance no longer increases uniformly with $\alpha$. In fact, for minimal circuit depth the relation reverses. This is not a surprise, as the signal size at short sequence lengths increases sharply with decreasing $\alpha$ (as seen in Figure \ref{fig:figure2}a). Computing the associated signal-to-noise ratio (not shown) restores the long sequence assessment that increasing $\alpha$ increases the sample complexity. These findings are also in agreement with the analytical findings for noiseless local Cliffords (which can be seen as the short sequence length limit) discussed in Appendix \ref{appendix:varnonoise}.

As we noted before, there is generally a trade-off between mixing length and the variance of the distribution being sampled from at each sequence length. This suggests there should be an optimal choice of $\alpha$ depending on the experiment's parameters, in particular the noise rates, circuit architecture and number of measurements, that minimizes the uncertainty in the extracted decay rate. Using the simulated data of \ref{fig:figure2}a and \ref{fig:figure2}b, which have no sampling error, we can perform a synthetic fitting procedure in the following way. We begin by calculating the mixing length for a fixed $\alpha$ by searching for the first sequence length where the $\alpha$ decay curve coincides (up to a negligible error) with the $\alpha=3$ curve. Then, considering only the points beyond this mixing length (a common practice, which we term "cautious fitting" and discuss further below), we can calculate the fundamental uncertainty in an estimator of the decay rate. In particular we consider a weighted version of the method detailed in section IV.A of Ref.~\cite{ep}. We can give an indication of the fitting uncertainty by linearizing the associated least-squares problem around the (known) solution and calculating fit uncertainty from the covariance.

We show the fitting uncertainty relative to $\alpha=1$ in Figure \ref{fig:figure2}c. 
This relative uncertainty is generally independent of the number of samples, so should give a good indication of the behavior of actual experiments.
For this particular setting, we find an optimal choice of $\alpha=1.5$ in-between the standard filter function and the trace filter function. 
The optimal filter function reduces the uncertainty by a factor of $4$,  compared to the standard filter function. 
Note that an a priori closed-form prediction of the optimal value of $\alpha$ is difficult, and becomes even harder when the signal itself is noisy (see below). As a rule of thumb we find that the higher the noise rates, the closer the optimal $\alpha$ is to $\alpha=3$ in our setting. 

\section{Cautious decay rate extraction in practice.}
In practice, determining the minimal mixing length from an RB signal with finite statistics is a nontrivial task. The noise-free mixing length, which can often be calculated exactly, can differ substantially from the mixing time in noisy conditions (see Figure \ref{fig:mix}a for an example). Hence the mixing length cannot be determined a priori, but must be estimated from experimental data.  
Appendix~\ref{appendix:fit} introduces three different methods for empirically estimating the mixing length based on bounding (i) the $\chi^2$ value of the residual statistics, (ii) the number of nontrivial eigenvalues of the autocorrelation matrix (singular values of the Hankel matrix), and (iii) approximating the experimental data with a tensor network heuristic using a simple noise model.  
For (iii) we determine the parameter of a local depolarizing noise channel by minimizing the $\chi^2$ distance of the evolution to the experimental data, and then determine a mixing length from the Hankel matrix (as in ii) but here using the tensor network evaluation without statistical noise. This can be thought of as a heuristic denoising step.
Among the mixing length heuristics considered, we found the tensor-network-based method to be the most reliable (even though it has the least formal justification). 
We use it to restrict the data to an estimated single-exponential regime before extracting the decay parameter, and refer to this procedure as \emph{cautious fitting}. 
In contrast, we call a single-exponential regression over the full data set, including the shortest circuit depths, \emph{optimistic fitting}. 

We illustrate that optimistic fitting indeed yields a (substantial) bias in the fitted decay rate decreasing with $\alpha$ with the following numerical simulation. For simplicity we use the Qiskit IBM Guadalupe simulator that is supposed to be a digital twin of a QPU with a built-in noise model unknown to us \cite{Qiskit}. We simulate the same brickwork circuit architecture of $2$--qubit Cliffords as before on 6 qubits using $1.2 \cdot 10^5$ circuits per length in $\{1, 2, \ldots, 50\}$ executed $100$ times each. 
The RB signal and second moments show the same behavior as the exact evaluation before (see Figure~\ref{fig:wide2} in the appendix). 
Figure~\ref{fig:figure2}d shows the extracted infidelities $\varepsilon$ with different values of $\alpha$ using the `cautious'
 and `optimistic' approaches.  
The standard filter function $\alpha=1$ leads to substantial bias in the optimistic estimate compared to the `cautious' method. The bias decreases for larger $\alpha$ and vanishes close to $\alpha=3$. 
The infidelities extracted from `cautious fits' are comparable for all values of $\alpha$.  
The setting underlines the importance of care when choosing the mixing length already for small system sizes. 
Note that using the abundant number of circuits and measurement statistics available in the small-scale simulation it is still possible to extract accurate and precise infidelity estimates for all values of $\alpha$.  
As we will see in the experiments, this changes for higher noise levels, larger system sizes and limitations in measurement statistics. Here already the reliable determination of the mixing length becomes challenging.

\providecommand{\panelplot}[5]{%
    \begin{tikzpicture}[baseline=(im.south)]
        \node[inner sep=0pt] (im) {\includegraphics[width=#1]{#2}};
        \node[anchor=#4, font=\bfseries, inner sep=1pt,
              fill=white, fill opacity=0.75, text opacity=1]
            at ([shift={#5}]im.#4) {#3};
    \end{tikzpicture}%
}

\newlength{\ensiconheight}\setlength{\ensiconheight}{4mm}
\providecommand{\ensicon}[1]{%
    \raisebox{-0.22\height}{\includegraphics[height=\ensiconheight,
        trim=#1,clip]{rb_paper_jm_v4.pdf}}%
}
\providecommand{\iconCl}{\ensicon{113.8 76.9 379.6 84.8}}   
\providecommand{\iconIco}{\ensicon{129.2 76.9 362.6 84.8}}  

\newlength{\panelcolheight}
\settoheight{\panelcolheight}{%
    \includegraphics[width=5cm]{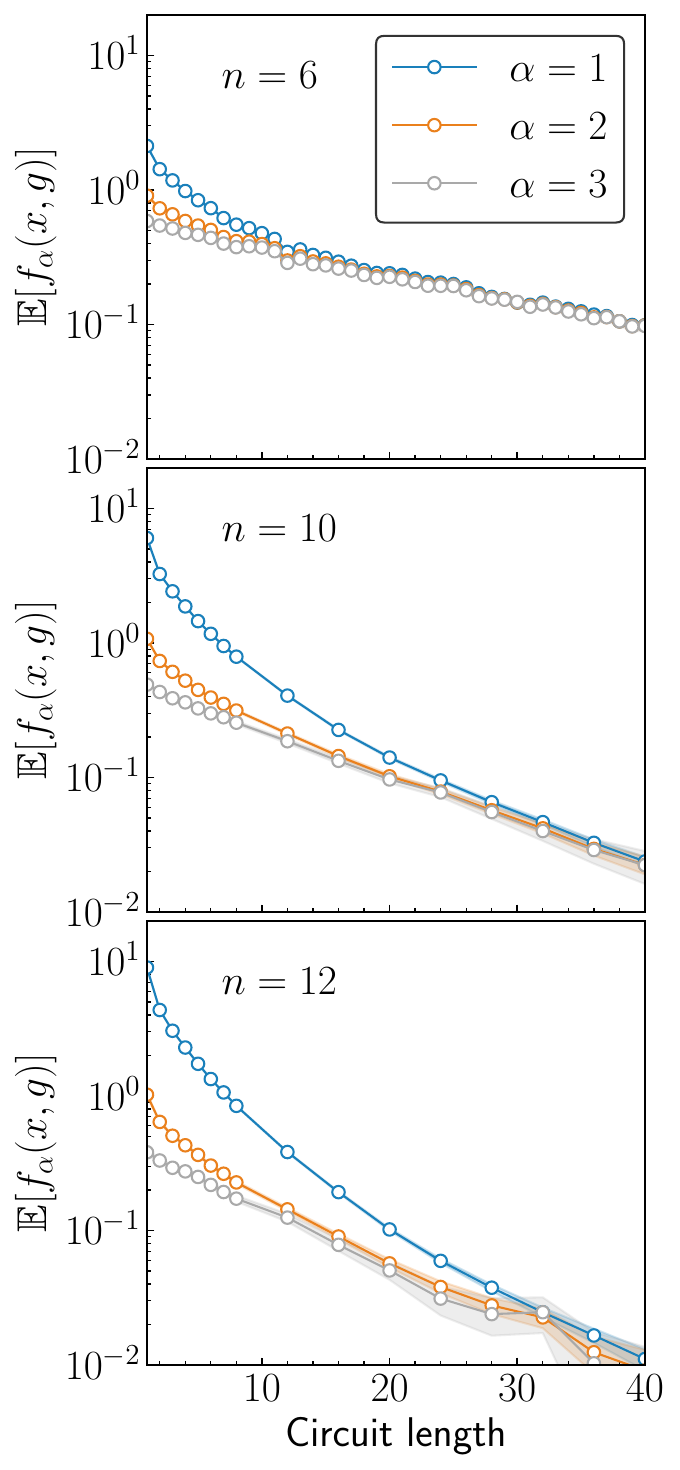}}

\begin{figure*}[htbp]
    \centering
    \hbox{%

        \begin{tikzpicture}[baseline=(current bounding box.north)]
        \hspace{1cm}%

            \node[inner sep=0pt] {%
                \begin{minipage}[t][9cm][s]{9cm}
                    \centering
                    \panelplot{10cm}
                        {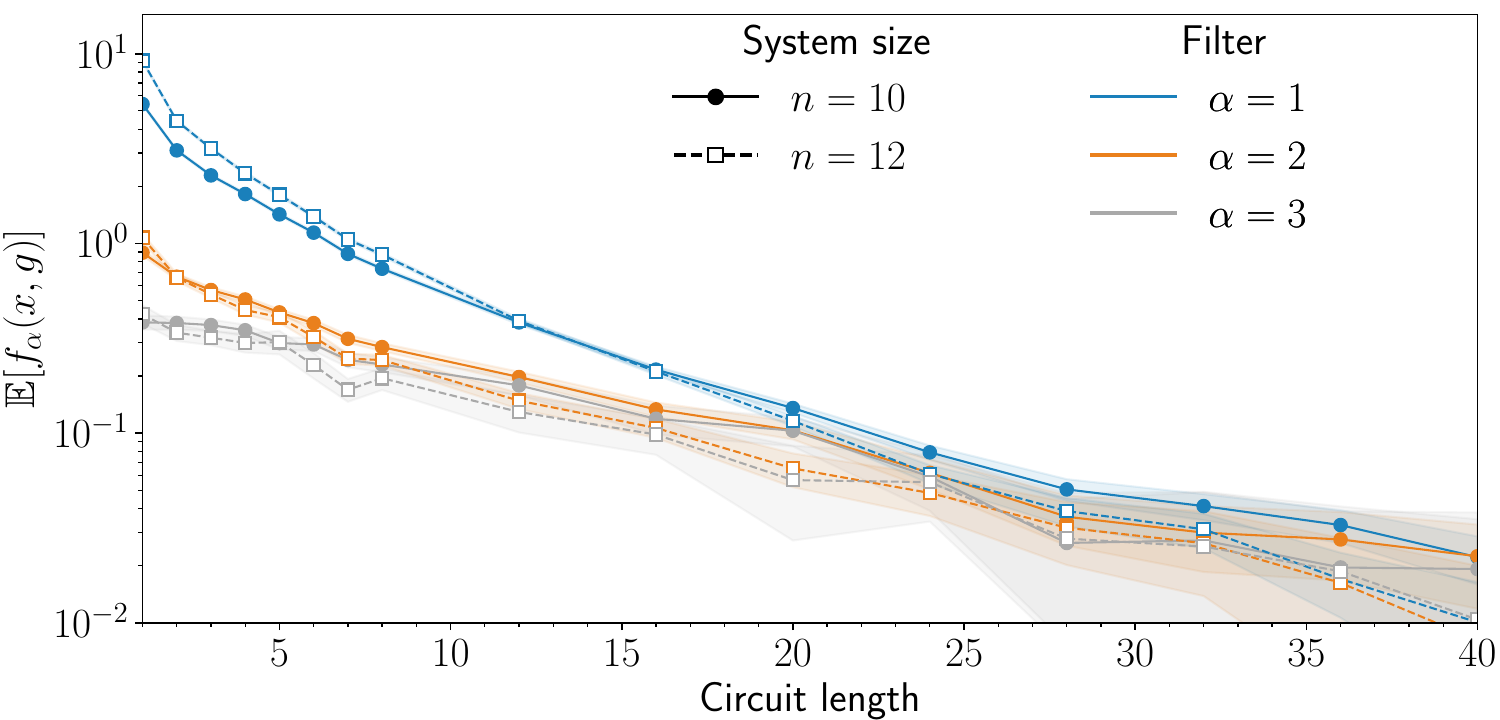}
                        {(a)\,\iconCl}{south west}{(1.10cm,0.77cm)}

                    \vfill

                    \makebox[10cm][c]{%
                        \panelplot{4.5cm}
                            {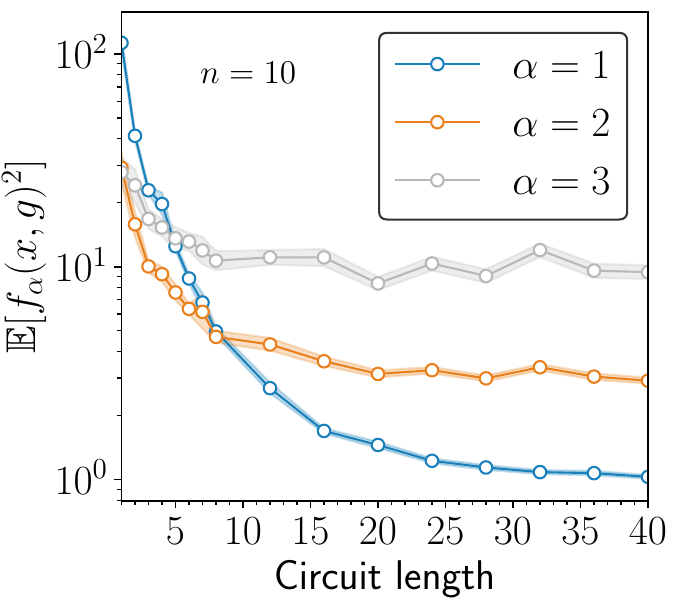}
                            {(b)\,\iconCl}{south west}{(0.92cm,0.78cm)}%
                        \hspace{0.8cm}%
                        \panelplot{4.5cm}
                            {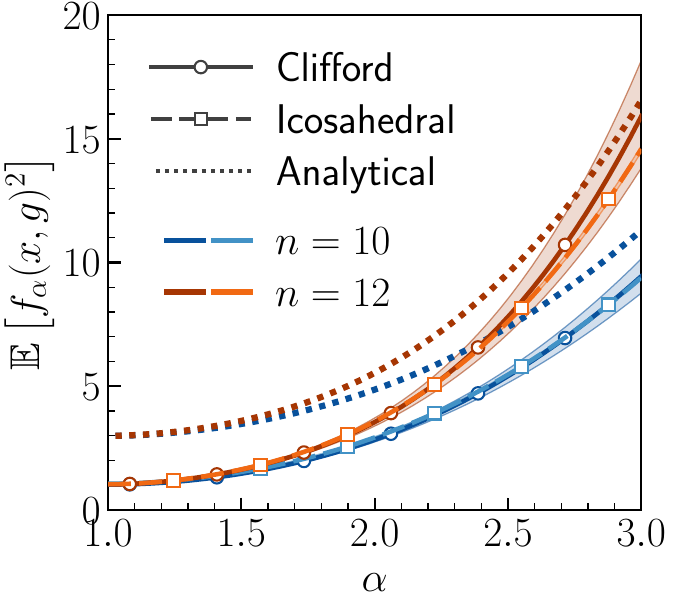}
                            {(c)}{south east}{(-0.30cm,0.68cm)}%
                    }
                \end{minipage}%
            };
        \end{tikzpicture}%
        \hspace{2.5cm}%
        \begin{tikzpicture}[baseline=(current bounding box.north)]
            \node[inner sep=0pt] {%
                \begin{minipage}[t]{4.2cm}
                    \vspace{0pt}
                    \centering
                    \panelplot{4.2cm}
                        {Icosahedral_curves_Rigetti.pdf}
                        {(d)\,\iconIco}{south west}{(0.92cm,0.77cm)}
                \end{minipage}%
            };
        \end{tikzpicture}%
    }
    \caption{Randomized benchmarking experiments performed on the Cepheus-1-36Q Rigetti device. We use random circuits as described in Figure \ref{fig:fig1_explanation}b and show results for both the Clifford ensemble and the icosahedral ensemble. For $n=6$, $m \in \{1\hdots 40\}$ and for $n=10/12$ we have $m \in \{1,  2,  3,  4,  5,  6,  7,  8, 12, 16, 20, 24, 28, 32, 36, 40\}$. For the number of circuits and shots, refer to Table \ref{tab:experiments}. Uncertainties are estimated using bootstrapping procedures and shaded areas represent 95\%-uncertainty intervals. In panel (a) we show the first moments using the Clifford ensemble and in panel (b) we show the corresponding second moments for $n=10$. In panel (c), we compare the estimated second moments for both ensembles at long sequence length ($m=40$) to the analytical second moments in the absence of noise (see Proposition~\ref{prop:var_no_noise} and Equation \ref{eq:varglob}). In panel (d), the RB curves for the icosahedral ensemble are shown.}

    \label{fig:fig3_rigetti_exp}

\end{figure*}

\section*{Experiment}

We demonstrate our method on Rigetti's superconducting 36-qubit quantum processor, Cepheus-1-36Q \cite{Field2024}. 
Due to the increase in variance with $\alpha$, precise infidelity estimates require the execution of many different random circuits. 
Increasing the number of shots a fixed circuit is executed with has traditionally been more efficient for runtime than a corresponding increasing in the number of distinct circuits executed. However, collecting more shots of the same circuit is often ineffective, in particular for Clifford circuits, in improving precision of observable estimation for small and intermediate system sizes~\cite{heinrich2022randomized,helsen2023thrifty}.
 We show here that we are able to execute low shot counts of many different circuits without a corresponding increase to the wall-time of the data acquisition by using an FPGA-integrated implementation of the random circuit sampling. 

Recall that the typical execution of a circuit on many platforms, including the Rigetti stack, 
follows a standard pattern.  The program written in a quantum instruction language must be compiled to instructions and configuration data suitable for the control system. This executable binary is then loaded onto the control system where it is usually interpreted by custom FPGA firmware, executed $S$ times and then the measurement results are returned to the user. For experiments where $S$ is large, this is a relatively efficient procedure---the time to compile and load the program and return the results will be small compared to the sampling time. However, in the scenario of many circuits with few shots, the overhead of loading the programs and returning results can become dominant, especially in hardware based on superconducting circuit based processors, due to their fast operation times. 

To avoid the per-circuit overhead, we utilized parameterized quantum circuits with the instructions to generate new sets of parameters on each shot embedded into the control system binaries themselves. The runtime re-writable parameters in these binaries include the angles of $\text{RZ}(\phi)$ gates. Our $\text{RZ}(\phi)$ gates are realized as virtual Z gates that adjust the phase of the RF drive of subsequent $\text{RX}(\theta \in k.\pi/2)$, $k \in \mathbb{Z}$, and do not otherwise alter the pulse timing schedule of the circuit \cite{McKay2017,Abrams2020}.
We utilized a 48-bit linear-feedback-shift-register (LFSR) pseudo-random number generator  developed previously for readout randomization \cite{arrasmith2023development}, masked to 16 bits at the output.
To translate this random integer into a random element from a discrete unitary subgroup, we loaded a lookup table of the angles $\{ \phi_0, \phi_1, \phi_2 \}$ onto the control system in their native 48-bit fixed-point format for each member of the unitary group decomposed to $\text{RZ}(\phi_2) \, \text{RX}(\pi/2) \, \text{RZ}(\phi_1) \, \text{RX}(\pi/2) \, \text{RZ}(\phi_0)$. 
For the $24$-element single--qubit Clifford group, this corresponds to a small $24 \times 3$ element lookup table which is shared by all operations on each control system unit. The random integer is then modulo-divided by the size of the lookup table to determine the selected operation. 
We provided each qubit with its own LFSR and seeded each qubit differently. 
In the post-processing we determine the executed gates from only the seed values that were sent to the control system omitting the need to store and communicate any circuits. See Figure~\ref{fig:fig1_explanation}d for an illustration. 

Using this technique, the throughput of the randomized circuit is limited primarily by the passive reset time, which was set to $200\mu s$ for this experiment.  
This allows us to execute $10^6$ different circuits in a few minutes, making accurate randomized benchmarking feasible for larger systems.

To limit the size of the lookup tables and avoid coordination of qubit controllers, the FPGA-integrated randomization is limited to independently randomized single--qubit gates. 
We construct the circuit by interleaving entangling layers of CZ gates with random layers of single--qubit gates; see Figure~\ref{fig:fig1_explanation}b. 
The construction is equivalent to a brickwork circuit of two--qubit bricks consisting of a CZ gate with independent random single--qubit gates applied before and after it on both qubits.

The resulting circuits are theoretically expected to have slightly different mixing behavior than the circuits with uniformly random two-qubit Clifford gates analyzed before, but for long sequence lengths converge to a $2$-design, as required for the filter functions in this paper to apply.
We implement three different single-qubit ensembles: the  $24$-element Clifford group, the $60$-element icosahedral group and the $12$-element tetrahedral group. The icosahedral group forms a $4$-design, which means that for this ensemble it is actually not ineffective to take multiple shots per circuit~\cite{helsen2023thrifty} (in contrast to the Clifford group, as discussed before). As a consequence, we choose to use $10$ shots per circuit for the icosahedral ensemble, while for the Clifford and tetrahedral ensembles we stick to the single-shot case. Repeating circuits reduces the number of filter function evaluations, which are particularly costly for the icosahedral group, where phase space simulation is not possible. 

We show the results of a $6$-qubit experiment using the icosahedral ensemble and a $10$- and $12$-qubit experiment with both ensembles in Figure \ref{fig:fig3_rigetti_exp}. The results for the tetrahedral group can be found in Figure \ref{fig:tetra}. The experimental settings and estimated infidelities for all experiments are summarized in Table \ref{tab:experiments}.

The empirical first and second moments show qualitatively similar behavior to the analytical evaluations and simulations discussed above. 
For the $6$-qubit data, we find an estimated minimal mixing length of $18$ using the model-based TN-simulation, meaning that cautious fitting is possible for all values of $\alpha$, similar to the simulations shown in Figure \ref{fig:figure2}. 

For the larger system sizes $n=10$ and $n=12$, the estimated mixing lengths for filter functions with $\alpha$ close to one approach, and in the $n=12$ case exceed, the maximum circuit depth of $40$ layers (e.g. for the icosahedral ensemble, we estimate mixing lengths of $36$ and $40$ respectively). Thus, accurate infidelity estimates require moving away from the standard filter function in this setting. By contrast, we find that the trace filter function, $\alpha = 3$, yields accurate and precise infidelity estimates already at very short circuit depths, below $10$ layers. 

We note that all three single-qubit groups yield comparable infidelities, rendering them qualitatively interchangeable benchmarks of this particular device.

In Figure~\ref{fig:fig3_rigetti_exp}c, we compare the second moments predicted by Proposition~\ref{prop:var_no_noise} and the second moments encountered in all experiments at long sequence lengths.
The converged second moments behave qualitatively similarly to the predicted idealized second moments, even for the non-Clifford groups,  with an offset in the constant mainly affecting small values of $\alpha$.

\newcolumntype{B}{!{\vrule width 0.5pt}}
\newcolumntype{E}{!{\vrule width \heavyrulewidth}}

\newcommand{\val}[2]{$#1\!\pm\!#2$}  
\newcommand{\nodata}{---}        
\newcommand{\tbd}{$?$}           
\newcommand{\unit}{\scriptsize$(\times 10^{5})$}  

\begin{table}[t]
\centering
\setlength{\tabcolsep}{3pt}
\renewcommand{\arraystretch}{1.1}
\setlength{\aboverulesep}{0pt}
\setlength{\belowrulesep}{0pt}
\setlength{\extrarowheight}{2pt}
\footnotesize
\begin{tabular}{E c B c B cc B cc B cc E}
\toprule
\multirow{3}{*}{$n$} &
 & \multicolumn{2}{c B}{\textbf{Tetrahedral}}
 & \multicolumn{2}{c B}{\textbf{Clifford}}
 & \multicolumn{2}{c E}{\textbf{Icosahedral}} \\
 & $N$ & $S$ & $\epsilon$ & $S$ & $\epsilon$ & $S$ & $\epsilon$ \\
 & \unit & & & & & & \\
\midrule
 6 & $2.5$ & $1$ & \val{0.045}{0.001}
        & \nodata & \nodata
        & $1$ & \val{0.047}{0.001} \\
10 & $1.0$ & $1$ & \val{0.08}{0.01}
        & $1$ & \val{0.08}{0.01}
        & $10$ & \val{0.082}{0.003} \\
12 & $1.0$ & $1$ & \val{0.14}{0.03}
        & $1$ & \val{0.10}{0.02}
        & $10$ & \val{0.103}{0.007} \\
\bottomrule
\end{tabular}
\caption{%
Overview of the randomized benchmarking experiments, arranged
by system size $n$. For each experiment, we
list the number of circuits $N$ per sequence length, the number of measurement shots $S$ per circuit and the infidelity
$\epsilon$
fitted with the trace filter ($\alpha = 3$). The uncertainty represents
a 68\%-uncertainty interval. A dash indicates that the
corresponding experiment was not performed.%
\label{tab:experiments}}
\end{table}

\section{Discussion}
Randomized benchmarking with structured random circuits faces a basic tension: experimentally convenient circuit ensembles may be far from uniformly distributed at short depths. This means the resulting RB signal can contain mixing-induced transients that obscure the fidelity decay of interest. Waiting for these transients to decay is often impractical, since it may require longer circuits, higher signal-to-noise ratio, or a larger measurement budget than the experiment allows. Thus, the mixing length becomes an experimental constraint.

In this work we showed that this constraint can be loosened in post-processing. 
We introduced a filter function family $f_\alpha$ interpolating between the standard LinXEB-type filter function and a trace filter function. We analytically establish that the parameter $\alpha$ controls the overlap of the RB signal with mixing-induced spectral components: under practical local symmetry assumptions, the trace filter function suppresses the ideal mixing contribution and exposes the fidelity decay at much shorter circuit depths. This bias reduction comes with increased variance and post-processing complexity, yielding a trade-off that can be tuned through $\alpha$. We also performed an in-depth analysis of the fitting problem itself, discussing the cautious fitting approach (which ought to be the standard approach in any RB/LinXEB decay experiment) and providing a comparative analysis of different methods for ascertaining the mixing length of a decay sequence (see appendix \ref{appendix:fit}). As discussed there, determining the mixing length is a nontrivial problem (and we caution against ``eyeballing'' the mixing length in an ad hoc fashion). Further work investigating the effect of noise on the mixing length would be very valuable in this respect and could also give more insights in the magnitude of the bias that is incurred when the use of sufficiently long sequences is infeasible. 

We note here that the general filtered RB protocol \cite{helsen2022general,heinrich2022randomized} is formulated for arbitrary (compact) groups using a representation-theoretic generalization of the filter function. For simplicity, we restricted our attention here to measures that converge to unitary $2$-designs, for which $f_1$ is the only nontrivial standard filter function. But an analogous representation-theoretic generalization of our results is possible and would be an interesting avenue of further research.

Beyond our analytical and numerical work, we also demonstrated our approach on superconducting qubits, extracting infidelities of one-dimensional entangling circuits in regimes inaccessible to the standard LinXEB-type protocol. This was enabled by a hardware-integrated implementation of random circuit sampling, which executes millions of distinct circuit instances near the physical throughput limit of the platform and removes control-stack latency as a dominant bottleneck. 

Several limitations remain. For larger systems, the variance of the filtered signal becomes the limiting factor, and further optimizing the allocation of circuit depths, circuit instances, and shots remains an important practical question. In addition, evaluating the filter function can be classically nontrivial. For the Clifford circuits, stabilizer methods make the computation feasible in the relevant regimes, but we showed the worst-case evaluation problem is hard and the intermediate-depth complexity of local Clifford circuits is not yet fully understood. 

More generally, our work demonstrates that scalable benchmarking with non-uniform circuit ensembles requires co-designing the filter function, fitting window, classical evaluation method, and hardware sampling architecture. Trace-filtered randomized benchmarking provides a new tool that extends experimental capabilities where short-depth bias was the main limitation.

\section{Acknowledgments}
M.H.~acknowledges funding by the Deutsche Forschungsgemeinschaft (DFG, German Research Foundation) under Germany’s Excellence Strategy - Cluster of Excellence Matter and Light for Quantum Computing (ML4Q) EXC 2004/1 - 390534769. J.W.~acknowledges financial support from the ERC (q-shadows 101117138). JH acknowledges funding from the Dutch Research Council (NWO) through a Veni grant (grant No.VI.Veni.222.331) and the Quantum Software Consortium (NWO Zwaartekracht Grant No.024.003.037). AK is funded through the  Groeifonds QDNL Cat-1 Program.
\paragraph*{Data and materials availability.} The experimental data is
available from the authors upon reasonable request.

\bibliography{apssamp}
\onecolumngrid
\appendix
\section{Appendices}
\setcounter{secnumdepth}{1}

\section{Weingarten calculus}
\label{appendix:wein}
The numerical arguments discussed in the main text are based on the Weingarten calculus \cite{collins2022weingarten}. In this appendix, we will discuss the mathematical basis of the Weingarten calculus. This discussion closely follows Ref. \cite{Haarintro}.
\begin{definition}
The Haar measure $\mu$ on a compact subgroup $G(2^n)$ of the unitary group $U(2^n)$ is the unique probability measure that is both left and right invariant under group multiplication, i.e. for all integrable functions $f$ and for all $V\in G(2^n)$, we have \begin{equation}
    \int_{G(2^n)}f(U)d\mu=\int_{G(2^n)}f(VU)d\mu=\int_{G(2^n)}f(UV)d\mu.
\end{equation}
\end{definition}
Given a Haar measure $\mu$, we can define moment operators on this measure.
\begin{definition}
    The k-th moment operator $\mathcal{M}_{\mu}^{(k)}: \mathcal{L}((\mathbb{C}^{2^n})^{\otimes k})\rightarrow \mathcal{L}((\mathbb{C}^{2^n})^{\otimes k})$ with respect to the Haar measure $\mu$ is defined as \begin{equation}
        \mathcal{M}_{\mu}^{(k)}(O)=\mathbb{E}_{U \sim \mu}[U^{\otimes k}OU^{\dagger \otimes k}]
    \end{equation}
    for all operators $O \in \mathcal{L}((\mathbb{C}^{2^n})^{\otimes k})$.
\end{definition}
For the evaluation of moment operators, it is important to understand commutants.
\begin{definition}
    The k-th order commutant of $S \subseteq \mathcal{L}(\mathbb{C}^{2^n})$ is defined as \begin{equation}
        \text{Comm}(S,k)=\{A \in \mathcal{L}((\mathbb{C}^{2^n})^{\otimes k}) : [A,B^{\otimes k}]=0 \quad  \forall \quad B\in S\}.
    \end{equation}
\end{definition}
It is well known that the k-th order commutant of the unitary group can be fully described in terms of permutation operators. 
\begin{definition}
    Given an element $\pi \in S_k$ of the symmetric group $S_k$, we define a permutation matrix $V_{2^n}(\pi)$ that permutes the $k$ copies of the $n$-qubit Hilbert space according to the permutation prescribed by $\pi$. That is, \begin{equation}
        V_{2^n}(\pi)|\psi_1\rangle\otimes\hdots \otimes |\psi_k\rangle=|\psi_{\pi^{-1}(1)}\rangle\otimes \hdots \otimes |\psi_{\pi^{-1}(k)}\rangle
    \end{equation}
    for all $|\psi_1\rangle, \hdots, |\psi_k\rangle \in \mathbb{C}^{2^n}$.
\end{definition}
The famous Schur-Weyl duality \cite{goodman2009symmetry} gives a direct expression of the k-th order commutant of the unitary group in terms of these permutation matrices.
\begin{theorem}
    The k-th order commutant of the unitary group is the span of the permutation matrices $V_{2^n}(\pi)$ associated with the symmetric group $S_k$, with $\pi \in S_k$. \begin{equation}
        \text{Comm}(U(2^n),k)=\text{span}(V_{2^n}(\pi):\pi \in S_k).
    \end{equation}
\end{theorem}
For a proof, please refer to Ref.\ \cite{goodman2009symmetry}. Note that the Clifford group is a unitary 3-design, which means that for $k\leq3$, the k-th order commutant of the Clifford group and the unitary group are identical. An important result in the theory of moment operators is that the k-th moment operator is the orthogonal projector onto the commutant. 
\begin{theorem}
    The moment operator $\mathcal{M}_{\mu}^{(k)}[\cdot]=\mathbb{E}_{U \sim \mu}[U^{\otimes k}[\cdot] U^{\dagger \otimes k}]$ is the orthogonal projector with respect to the Hilbert-Schmidt inner product onto the commutant \text{Comm}($U(2^n),k$). 
\end{theorem}
The corresponding proof can be found in Ref.\ \cite{Haarintro}. It is important to note that the permutation matrices do not necessarily form an orthonormal basis of the commutant. Expressing a moment operator in terms of such a non-orthonormal basis of permutation matrices results in \begin{equation}\label{eq:wein}
    \mathbb{E}_{U \sim \mu}[U^{\otimes k}O U^{\dagger \otimes k}]=\sum_{\pi\in S_k,\sigma\in S_k}W(\pi^{-1}\sigma,2^n)\text{Tr}\left( V_{2^n}^{\dagger}(\sigma)O\right)V_{2^n}(\pi),
\end{equation}
where $W(\pi^{-1}\sigma,2^n)=G_{\pi,\sigma}^+$ are the Weingarten coefficients, defined in terms of matrix elements of the pseudoinverse of the Gram matrix, with matrix elements \begin{equation}
    G_{\pi,\sigma}=\text{Tr}\left(V_{2^n}(\pi)^{\dagger}V_{2^n}(\sigma)\right).
\end{equation}
Often, it is more convenient to use the Liouville representation \cite{blum2012density}. The moment operator $\mathcal{M}_{\mu}^{(k)}[\cdot]$ in the Liouville representation (which we denote with a hat) can be represented as \begin{equation}\label{eq:weinvec}
    \hat{\mathcal{M}}_{\mu}^{(k)}=\mathbb{E}_{U\sim \mu}U^{\otimes k}\otimes(U^*)^{\otimes k}=\sum_{\pi\in S_k,\sigma \in S_k}W(\pi^{-1}\sigma,2^n)|V_{2^n}(\pi))( V_{2^n}(\sigma)|,
\end{equation}
where $|X)$ indicates the vectorization of $X$.

\section{Signal and variance guarantees}\label{appendix:signal}
In this section we discuss in more detail propositions \ref{prop: bias} and \ref{prop:var_no_noise}) considered in the main text. 
We will work in the vectorized notation $|\cdot )$ for operators (in particular denoting by $|E_x)$ the vectorization of the operators $\ket{x}\!\bra{x}$ which project onto computation basis states).
The filter functions are parameterized by $\alpha$ and given by 

\begin{equation}
    f_{\alpha}(x,g)=(2^n+1)(x|P\omega(g)|\xi(\alpha)),
\end{equation}
with $\xi=\frac{1}{(1+\alpha)^n}(\mathbb{I}+\alpha Z)^{\otimes n}$,  and $\text{Tr}(\xi)=(\frac{2}{1+\alpha})^n$. With $\omega(g)$ we denote the matrix representation of random brickwork circuit $g$. The projector on the adjoint irrep $P$ is given by $P=I-\frac{1}{2^n}| \mathbb{I})(\mathbb{I}|$.

To formally state our guarantees we need to formalize several assumptions we make on the device implementing the quantum operations. We do this using the so-called implementation map, denoted by $\phi$, which maps elements of the group to CPTP maps. It can be though of as modeling the noisy implementation of the ideal operations (specified by $\omega$). In our context, $\phi$ is only required to exist on the support of $\nu$. 
Aligned with standard conventions, we say that an experiment admits a \emph{time-stationary and Markovian gate-set description} 
if there exists a 
a state $\tilde \rho$, POVM effects $\tilde E_x$ (which might differ from their ideal implementations) and a channel $\phi(g)$ associating quantum channels to unitary instructions,  such that executing the sequence of instructions $g = (g_1, \ldots, g_m)$ leads to sampling from the Born probability
\begin{equation}
    p(x,g) = (\tilde E_x| \phi(g_m) \ldots \phi(g_1) | \tilde \rho)\,. 
\end{equation}

The main restriction of this model is that we assume that the \emph{same} quantum channel models the noisy-implementation of the instruction $g$ in the support of $\nu$ independent of its position in the sequence and the time of the execution of the particular sequence.  We are now in the position to state our main result on the trace filter function.

\subsection*{Signal}
In an RB experiment, we are effectively approximating $\mathbb{E}_{g \sim \nu^{*m}}\left[\sum_x p(x,g) f_{\alpha}(x,g)\right]$, where $\nu^{*m}$ is the effective measure imposed by taking a random brickwork circuit of length $m$ consisting of 2-qubit Haar random Cliffords. Using Weingarten calculus, we can evaluate this expectation value. 

\begin{align}\label{eq:exp}
\mathbb{E}_{g \sim \nu^{*m}}\left[\sum_x p(x,g) f_{\alpha}(x,g)\right]=
 (2^n+1) \int_{g_1,\ldots,g_m} \sum_{x\in\{0,1\}}^n (E_x\otimes \widetilde{E}_x| \omega(g_m\cdots g_1)\otimes (\phi(g_m)\cdots \phi(g_1)|\xi \otimes \rho ) d\nu^{*m} - \frac{2^n+1}{2^n}\text{Tr}(\xi).
\end{align}

We now argue that if the noisy implementation $\phi$ does not break the local symmetry of a perfect implementation $\omega$ (we will specify exactly what we mean by this), the 
RB signal of the trace filter contains an explicit projection onto the dominant signal subspace. This argument was first sketched by a subset of the authors in the preprint \cite{heinrich2022randomized}.\\

We begin by formulating an intermediate lemma. We denote the single-qubit Clifford group as $\operatorname{Cl_1}$ and understand $\omega$ as representing $g_1, \ldots, g_n \in \operatorname{Cl_1}^{\times n}$ as the product unitary channel $X \mapsto g_1 \otimes \ldots \otimes g_n  X g^\dagger_1 \otimes \ldots \otimes g^\dagger_n$.  We say a superoperator $\mathcal X$ exhibits \emph{local Clifford symmetry} if $[X, \omega(g)]  = X\omega(g) - \omega(g)X = 0$ for all  $g \in \operatorname{Cl_1}^{\times n}$. For such superoperators can prove the following lemma relating the trace filter function "observable" $\xi$ to the superoperator trace.  
\begin{lemma}\label{lem:trace}
Consider a superoperator $\mathcal X$ with the property $[\mathcal X, \omega(g)] = 0$ for all $g \in \operatorname{Cl_1}^{\times n}$.  
Then, 
\begin{equation}
   (\xi(3)| \mathcal X | 0) =  4^{-n} \operatorname{Tr}(\mathcal X) \,,
\end{equation}
with $\xi(3)=\frac{1}{4^n}(\mathbb{I}+3 Z)^{\otimes n}$.
\end{lemma}

\begin{proof}
Let $w((q,p))$ denote a multi-qubit Pauli operator indexed via the discrete phase-space vector $(q,p) \in \mathbb F^{2n}_2$.  
Recall the following properties: 
First, $2^{-n/2} w((p,q)) $ constitutes an orthogonal basis. Second, $w((q, 0))$ are the multi-qubit Pauli-$Z$ operators and   
\begin{equation}
     |0 \rangle\! \langle 0 | = 
     2^{-n} \sum_{q \in \mathbb F^n_2} w((q, 0))\,. 
\end{equation}
Third, for all $p \in \mathbb F_2^n$ there exist $g \in \operatorname{Cl}^{\times n}_1$ such that $\omega(g)(w((q,p))) = w((q, 0))$ and the orbit of $w((q, 0))$ has cardinality $3^{|q|}$ with $|q|$ the Hamming weight. 
Fourth, the commutant of the Pauli-group (discrete Heisenberg-Weyl group) are Pauli-diagonal maps, i.e. that are spanned by $|w((q,p))) (w((q,p))|$. Since the 
Pauli-group is a subgroup of $\operatorname{Cl}^{\times n}_1$ also $\mathcal X$ is diagonal in Pauli-basis. Combining these four properties, we find the lemma's statement by direct calculation  
\begin{equation}
\begin{split}
    \operatorname{Tr}(\mathcal X) &= 2^{-n} \sum_{a \in \mathbb F^{2n}_2} ( w(a) | \mathcal X | w(a))  \\
    &= 2^{-n} \sum_{z \in \mathbb F^{n}_2} 3^{|z|} ( w((z,0))|\mathcal X | w((z,0))) \\
    &= 2^{-n} \sum_{z,z' \in \mathbb F^{n}_2}  
    3^{|z|} ( w((z,0))|\mathcal X | w((z',0))) \\
    &= ( (\mathbb I + 3Z)^{\otimes n} | \mathcal X | 0) = 4^n ( \xi(3) | \mathcal X | 0)\,. 
\end{split}
\end{equation}
\end{proof}
Next we will derive our guarantees on the trace filter RB signal which are given as Proposition \ref{prop: bias} in the main text.  Given the shape of the filter function and the assumptions made on $\phi$ the actual derivation is reasonably straightforward. 
\begin{proposition}
Assume the RB experiment admits a time-stationary, Markovian gate-set description with perfect state-preparation ($| 0\rangle\!\langle 0|$, $\phi$, $\{\tilde E_i\}$) on the support of $\nu$.  Let $\tilde M = \sum_x |E_x ) ( \tilde E_x|$ and $T : X \mapsto \int \omega^\dagger(g) X \phi(g) \rm d\nu(g)$. 
Further assume that $[T^m(\tilde M), \omega(\operatorname{Cl}_1^{\otimes n})]=0$,. 
Then, 
\begin{equation}
     \mathbb E [f_3] = \frac{2^n + 1}{4^n} \operatorname{Tr}[P T^m(\tilde M)]\, .
\end{equation}
\end{proposition}
\begin{proof}
Under the assumption, the statement is an immediate consequence of Lemma~\ref{lem:trace} since $P$ is locally Clifford-invariant  
\begin{equation}
    \mathbb E [f_3(x,g)] = (2^n + 1)
    (\xi(3) |P\,T^m(\tilde M) | 0) 
    = \frac{2^n + 1}{4^n} \operatorname{Tr}[P\, T^m(\tilde M)]\,. 
\end{equation}
\end{proof}
If $\phi(g) = \omega(g)$, i.e.\  gates are perfectly implemented, 
and $\nu$ is right-invariant under $\operatorname{Cl}_1^{\times n}$, $T^m[\tilde M]$ commutes with local Clifford channels for all $m$ and operators $\tilde M$.  
Here, the first circuit layer implements (at least) a perfect local Clifford twirl. 
Further, for $\phi(g) = \omega(g)$,  
$P$ is the non-trivial fixed point of $T$ and the sub-dominant decays are associated to orthogonal subspaces.   As a consequence, for perfect gates the 
expected RB signal becomes $\mathbb E [f_3(x,g)] = \frac{2^n + 1}{4^n}\operatorname{Tr}[P\, T(\tilde M)]$ independent of $m$.

\subsection*{Variance}\label{appendix:varnonoise}

We now present analytical bounds on the variance induced by the interpolating filter function family (so $\alpha \in [1,3]$) for two extreme cases for the distribution $\nu$, namely (1) uniformly random $n$-qubit Clifford gates and (2) independently randomly sampled single qubit Clifford gates. In the context of brickwork circuits these can be seen as proxies for the long and short sequence length regimes respectively. We will present bounds for noiseless gates, i.e. $\phi = \omega$ and zero SPAM (i.e. $\tilde E_x = E_x$ and $E_0 = \rho$), but expect these to hold more generally (at least in a qualitative sense).\\

The variance of the RB signal can be determined from the first moment (calculated above) and the second moment of the distribution.
\begin{align}\label{eq:var}
\mathbb{E}_{g \sim \nu^{*m}}\left[\sum_x p(x,g) f^2_{\alpha}(x,g)\right]
= & \frac{(2^n+1)^2}{4^n}\text{Tr}(\rho)^2\text{Tr}(\xi)^2\underbrace{\int_{g_1,\ldots,g_m}\sum_{x\in \{0,1\}^n}(\widetilde{E}_x|(\phi(g_m)\cdots \phi(g_1)|\rho) d\nu^{*m}}_{\text{1st-order moment operator}} \\& \nonumber- 2\frac{(2^n+1)^2}{2^n}\text{Tr}(\rho)\text{Tr}(\xi)\underbrace{\int_{g_1,\ldots,g_m}\sum_{x\in \{0,1\}^n}(E_x\otimes \widetilde{E}_x| \omega(g_m\cdots g_1)\otimes (\phi(g_m)\cdots \phi(g_1)|\xi \otimes \rho )d\nu^{*m}}_{\text{2nd-order moment operator}}
\\& \nonumber+ (2^n+1)^2 \underbrace{\int_{g_1,\ldots,g_m}\sum_{x\in \{0,1\}^n}({E}_x^{\otimes 2}\otimes \widetilde{E}_x| \omega(g_m\cdots g_1)^{\otimes 2}\otimes (\phi(g_m)\cdots \phi(g_1)| \xi ^{\otimes 2}\otimes \rho)d\nu^{*m}}_{\text{3rd-order moment operator}}.
\end{align}

We need to evaluate the moment operators in Equation \ref{eq:var}. In the noiseless case, it is sufficient to analyze the case where $m=1$. We begin with uniformly sampled Clifford case.  We will use that, for group elements $g$ drawn from the Haar measure $\mu$ of $n$-qubit Cliffords, we have (using local invariance to remove the $x$ dependence)

\begin{equation}\label{eq:weinclif}
    \int ( \rho^{\otimes k}|\omega(g)^{\otimes k}| \rho \otimes \xi^{\otimes k-1}) d\mu=\frac{1}{2^n}\frac{1}{(1+\alpha)^{n(k-1)}}\sum_{\pi \in S_k}( V_{2^1}(\pi)|(\mathbb{I}+Z)\otimes(\mathbb{I}+\alpha Z)^{\otimes k-1}) ^{ n} \prod_{i=0}^{i=k-1}\left(\frac{1}{2^n+i}\right).
\end{equation}
This equation follows from Equation \ref{eq:weinvec} and the fact that \begin{align}
    &( \rho|V_{2^n}(\pi))=1 \\
    &V_{2^n}(\pi)=(V_{2^1}(\pi))^{\otimes n}
\end{align} for all $\pi$ and that \begin{equation}
    \sum_{\pi\in S_k,\sigma \in S_k}W(\sigma^{-1}\pi,2^n)( V_{2^n}(\pi)|=\sum_{\pi\in S_k}\prod_{i=0}^{i=k-1}\left(\frac{1}{2^n+i}\right)( V_{2^n}(\pi)|.
\end{equation}

Using Equations \ref{eq:weinclif},\ref{eq:var}, \ref{eq:exp} we can derive that \begin{align}
     \mathbb{E}_{g \sim \nu^{*m}}\left[\sum_x p(x,g) f_{\alpha}(x,g)\right]=&1-\frac{1}{(1+\alpha)^n} \label{eq:meanglob}\\
   \mathbb{E}_{g \sim \nu^{*m}}\left[\sum_x p(x,g) f^2_{\alpha}(x,g)\right]= \nonumber &(2^n+1)\left(\frac{1}{1+\alpha}\right)^{2n}\left(\frac{2-2^n+2^n(1+\alpha^2)^n+2(1+\alpha)^{2n}-4(1+\alpha)^n}{2^n+2}\right)\\ \label{eq:varglob}=&  \mathcal{O}\left(2+2^n\frac{(1+\alpha^2)^n}{(1+\alpha)^{2n}}\right)=\mathcal{O}\left(2+\left[1+\left(\frac{\alpha-1}{\alpha+1}\right)^2\right]\right).
\end{align}
which corresponds to the variance expression discussed in the main text.\\

For the local Clifford group ($Cl_1^{\times n}$) we can perform a very similar (and standard) calculation) to obtain variances for arbitrary $\alpha$. It is noteworthy that in this regime the variance (dominated by the third moment term) is actually smaller for $\alpha=3$ than it is for $\alpha =1$. We have 
\begin{align}\label{eq:single_q_var}
    &\mathbb{E}_{g \sim \nu^{*m}}\left[\sum_x p(x,g) f^2_{1}(x,g)\right]=\mathcal{O}\left(2^n\right) \\
    &\mathbb{E}_{g \sim \nu^{*m}}\left[\sum_x p(x,g) f^2_{3}(x,g)\right]=\mathcal{O}\left(\left(\frac{3}{2}\right)^n\right).
\end{align}
However this does not mean that the associated sample complexity is smaller. This reversal in variance is more than erased by the fact that in the short sequence length regime the expected value of the $\alpha=1$ is exponentially large in $n$ (this is a well known phenomenon in the LinXEB literature).

\begin{align}
    &\mathbb{E}_{g \sim \nu^{*m}}\left[\sum_x p(x,g) f_{1}(x,g)\right]=\mathcal{O}\left(\left(\frac{4}{3}\right)^n\right) \\
    &\mathbb{E}_{g \sim \nu^{*m}}\left[\sum_x p(x,g) f_{3}(x,g)\right]=\mathcal{O}\left(1\right).
\end{align}
One can calculate the signal to noise ratio in both cases by computing the ratio of the standard deviation and the expected value, resulting in a scaling of $\sqrt{\frac{9}{8}}^n$ for $\alpha=1$ and $\sqrt{\frac{3}{2}}^n$ for $\alpha=3$.

\section{Simulating RB using Weingarten calculus}\label{appendix:signal_sim}
In this section we use Weingarten calculus to directly calculate (numerically) expectation values and variances over the circuits encountered in RB experiments using brickwork random Clifford circuits. Specifically, we are interested in the expectation value and the variance of the filter functions when performing $n$-qubit RB with a brickwork circuit consisting of 2-qubit Haar random Cliffords. 
Although the measure $d\nu^{*m}$ is not a Haar measure, the brickwork circuits do consist of Haar distributed Cliffords and the moment operator with respect to $d\nu^{*m}$ can be expressed as a collection of local moment operators with respect to Haar measures, enabling us to use Equation \ref{eq:wein} (provided that the noise in the quantum circuit is modeled by a noise-channel that acts locally and independently on each local gate, i.e. the noisy implementation of a random 2-qubit Clifford $U$ is given by $\phi(\omega(U))=\Lambda\omega(g)$ with $\Lambda$ independent of $U$). Tensor network methods to evaluate such integrals over $d\nu^{*m}$ have been discussed in Refs. \cite{Haarintro,TNs,turkeshi2026lecture}. As an example, we will show how to evaluate \begin{equation}\label{eq:example}
    \int_{g_1,\ldots, g_m} d\nu(E_0 \otimes E_0 |\omega(g_m\cdots g_1)\otimes \phi(g_m)\cdots \phi(g_1)|\xi \otimes \rho ) ,
\end{equation} for $n=6$. 

 Leveraging Equation \ref{eq:weinvec}, we can rewrite Equation \ref{eq:example} in tensor network form. For simplicity, we assume that the noise can be modeled by qubit-independent, single-qubit noise, i.e. the noise channel on an $n$-qubit system $\mathcal{E}_n[\cdot]$ can be written as $\mathcal{E}_n[\cdot]=\mathcal{E}_1^{\otimes n}[\cdot]$, with $\mathcal{E}_1$ the noise channel experienced by an individual qubit. We denote the transfer matrix corresponding to the channel $I \otimes \mathcal{E}_1$ as $\Lambda_1$ and the transfer matrix corresponding to the channel $I \otimes \mathcal{E}_1 \otimes I \otimes \mathcal{E}_1$ as $\Lambda_2$. Furthermore, we define $\gamma \equiv \frac{1}{4(1+\alpha)^2}|I+\alpha Z)\otimes |I+Z) \otimes |I+\alpha Z)\otimes |I+Z)$ and $\xi=\frac{1}{16}\Lambda_2^{\dagger}|I+Z)^{\otimes 4}$. Then Equation \ref{eq:example} can be depicted in tensor network notation as
\begin{equation}
\begin{quantikz}[row sep=0.2cm, column sep=0.1cm]
\gate[2]{\xi} & \qw & & & \gate[2,style={rounded corners}]{\hat{\tau}_2^{(2)}}& \hdots\hdots\hdots && \qw & \gate[style={rounded corners}]{\Lambda_1} & \gate[2,style={rounded corners}]{\hat{\tau}_2^{(2)}} & \gate[2]{\gamma} \\
&  & \gate[2,style={rounded corners}]{\hat{\tau}_2^{(2)}} & \gate[style={rounded corners}]{\Lambda_1} & &\hdots\hdots\hdots & \gate[style={rounded corners}]{\Lambda_1}& \gate[2,style={rounded corners}]{\hat{\tau}_2^{(2)}} & \gate[style={rounded corners}]{\Lambda_1} & \qw &  \\
 \gate[2]{\xi} & \qw & & \gate[style={rounded corners}]{\Lambda_1} & \gate[2,style={rounded corners}]{\hat{\tau}_2^{(2)}}& \hdots\hdots\hdots & \gate[style={rounded corners}]{\Lambda_1}& \qw & \gate[style={rounded corners}]{\Lambda_1} & \gate[2,style={rounded corners}]{\hat{\tau}_2^{(2)}} & \gate[2]{\gamma} \\
&  & \gate[2,style={rounded corners}]{\hat{\tau}_2^{(2)}} & \gate[style={rounded corners}]{\Lambda_1} & & \hdots\hdots\hdots &\gate[style={rounded corners}]{\Lambda_1}& \gate[2,style={rounded corners}]{\hat{\tau}_2^{(2)}} & \gate[style={rounded corners}]{\Lambda_1} & &  \\
\gate[2]{\xi} & \qw &  &  \gate[style={rounded corners}]{\Lambda_1} & \gate[2,style={rounded corners}]{\hat{\tau}_2^{(2)}} & \hdots\hdots\hdots &\gate[style={rounded corners}]{\Lambda_1}& \qw & \gate[style={rounded corners}]{\Lambda_1} & \gate[2,style={rounded corners}]{\hat{\tau}_2^{(2)}} & \gate[2]{\gamma} \\
&  & \qw & \qw & & \hdots\hdots\hdots & &\qw & \gate[style={rounded corners}]{\Lambda_1} & \qw &
\end{quantikz}.
\end{equation}
Here $\hat{\tau}_n^{(k)}$ is a tensor that represents the effective action of the local projector onto the commutant $\text{Comm}(U(2^n),k)$, described by Equation \ref{eq:weinvec}. To derive ${\hat{\tau}_2^{(2)}}$, we first choose a basis to work in (see Ref.\ \cite{TNs} for a more detailed derivation of ${\hat{\tau}_2^{(2)}}$). A vectorized orthonormal basis of $\text{Comm}(U(2^2),2)$ is \begin{equation}
    \big{\{} \frac{1}{4}|\text{I})|\text{I}),\frac{1}{4\sqrt{15}}\left(|S)|\text{I}) +|\text{I})|S) +|S)|S) \right)\big{\}},
\end{equation}
with $S=X\otimes X+Y \otimes Y +Z\otimes Z$ for Pauli matrices $X,Y,Z$.
Inspired by this we choose the local basis \begin{equation}
    \big{\{}|\text{I})|\text{I}),|S)|\text{I}),|\text{I})|S),|S)|S) \big{\}}.
\end{equation}
Using Equation \ref{eq:weinvec}, we see that 

\begin{equation}
    \hat{\tau}_2^{(2)}=\frac{1}{5}\begin{bmatrix}
5 & 0 & 0 &0 \\
0 & 1& 1& 3\\
0 & 1& 1 & 3\\
0 & 1& 1& 3\\
\end{bmatrix}.
\end{equation}
To evaluate the tensor network, $\xi,\Lambda_1,\gamma$ (or at least their orthogonal projections onto this basis) must also be represented in this basis and then the network can be contracted using standard tensor network libraries. For larger qubit systems, it is necessary to use a matrix product states (MPS) representation. The bond dimension of the MPS representation will remain small throughout the tensor network evaluation, enabling efficient simulations (see Proposition 2 of Ref. \cite{TNs}).

\section{Complexity of evaluating the filter function}\label{appendix:filters}
In this appendix we provide further details on the computational complexity of evaluating the filter function $f_\alpha$. We begin consider the case where the gates used in the RB protocol are all Clifford. 
We will see that, even when all gates are Clifford, exact evaluation can be computationally difficult in the worst case. At the same time, the special structure of Clifford circuits allows for efficient evaluation in several important regimes.

If the circuit is Clifford, we can efficiently describe the state appearing in the filter function as a stabilizer state $\ket{S} = C^\dagger \ket{x}$ where $x$ is the measured bitstring and $C$ is the implemented Clifford operation. 
We are now interested in computing the quantity 
\begin{equation}
f_\alpha(x,C) = \frac{(2^{n}+1)}{(1+\alpha)^n} \sum_{w \in \{0,1\}^n\;: w\neq 0} \alpha^{|w|} \bra{S}Z^w \ket{S}.
\end{equation}
The state $\ket{S}$ is efficiently described by an extended stabilizer tableau
\begin{equation}
\mathrm{Tableau}(\ket{S})  = (s, Z, X),
\end{equation}
with the rows of $X,Z\in \{0,1\}^{n\times n}$ encoding the $X$ and $Z$ components of the stabilizer generators and $s\in\{0,1\}^{n}$ encoding the associated phases. The Pauli operators $Z^w$ has nonzero expectation value under $\ket{S}$ if only if $\pm Z^w$ is an element of the stabilizer group of $\ket{S}$. 
This is true only if there exists a vector $y\in \{0,1\}^n$ such that $y (Z, \, X) = (w,0) $ (i.e. we have to make the $X$ component zero). Finding all such $y$ (since we are interested in arbitrary $w$) is equivalent to finding a basis  for the kernel of $X$ (i.e. we want to solve $yX = 0$). Let's call the dimension of the kernel $k$ and denote by $V\in \{0,1\}^{k\times n} $ the matrix who's rows encode a basis for the kernel of $X$ (this can be found efficiently using Gaussian elimination). We can then write all such $y$ as $y=tV$ with $t\in \{0,1\}^k$. 
The filter function can now be evaluated as
\begin{equation}\label{eq:filter_evaluation}
f(x,C) = \frac{(2^{n}+1)}{(1+\alpha)^n}\sum_{t\in\{0,1\}^k \;: t\neq 0 } \alpha^{|tV|} (-1)^{|tV s|}\,.
\end{equation}
Direct evaluation of this equation requires calculating $2^k - 1$ terms, where $k = \operatorname{dim\,ker} X \leq n$.  
With a streaming implementation this costs $O(kn 2^k)$ time and $O(n^2)$ memory. In the worst-case $k = n$, and hence takes time exponential in the number of qubit.\\

Moreover, if we set $s$ to zero calculating the above reduces to calculating the weight enumerator of the linear code $V$ (thinking of its rows as codewords). Unless $\alpha = 1$ (or $\alpha = 0$) this is a $\#P$-hard problem, so we cannot expect a much better polynomial-time exact algorithm in $k$ for computing the filter function for arbitrary $\alpha$ than the exponential time direct calculation given above. \\

However we are not per se interested in the worst case complexity of this problem. For a uniformly random global Clifford unitary $C$, the unphased stabilizer of $C^\dagger\ket x$ is a uniformly random Lagrangian subspace $L \subset \mathbb F^{2n}_2$. 
The kernel dimension $k = \operatorname{dim}(L \cap Z_0)$, where $Z_0$ is the fixed Lagrangian subspace of $Z$-type Paulis. 
We have $2^k = | L \cap Z_0|$ and, thus, by linearity of expectation, 
\begin{equation}
  \mathbb E [2^k] = \sum_{z\in Z_0} \mathbb P[ z \in L]\,. 
\end{equation}
The zero vector contributes one. 
For any nonzero $z$, symmetry gives 
\begin{equation}
	\mathbb P[z\in L] = \frac{2^n - 1}{2^{2n} - 1} = \frac1 {2^n + 1}\, , 
\end{equation}
since each Lagrangian contains $2^n - 1$ non-trivial vectors of the $2^{2n}- 1$ non-trivial vectors in phase space.  
Therefore, 
\begin{equation}
  \mathbb E[ 2^k] =  1 + (2^n - 1) \frac1 {2^n + 1} = \frac{2^{n+1}}{2^n + 1}\, .
\end{equation}
Consequently, the expected number of nonvanishing summands in Eq.~\eqref{eq:filter_evaluation} is 
\begin{equation}
  \mathbb E [2^k - 1] = \frac{2^n - 1}{2^n + 1} \in O(1)\,. 
\end{equation}
Thus, for Haar-random global Clifford unitaries, direct evaluation is typically efficient, even though the worst-case and structured-circuit complexity remain exponential. 

Our work is mostly concerned with the regime where a quantum circuit has not yet approximated a Haar-random Clifford unitary. 
The other extreme are independent single-qubit Clifford rotations,
\begin{equation}
C = C_1\otimes \cdots \otimes C_n ,
\end{equation}
where each $C_i$ is sampled uniformly from the single-qubit Clifford group. 

For a product of single-qubit Clifford rotations, the state $C^\dagger \ket x$ is a product stabilizer state. 
The kernel dimension is distributed as $k \sim \operatorname{Binomial}\left(n,\tfrac{1}{3}\right)$
and $\mathbb E[k] = \tfrac{n}{3}$.

Hence, 
\begin{equation}
\mathbb E[2^k]=
\sum_{j=0}^{n}
2^j
\binom{n}{j}
\frac{1}{3^j}
\left(\frac{2}{3}\right)^{n-j}
=
\sum_{j=0}^{n}
\binom{n}{j}
\left(\frac{2}{3}\right)^{n}
=\left(\frac{4}{3}\right)^n\,.
\end{equation}

Thus, for a product of independent single-qubit Clifford rotations, the number of terms in Eq.~\eqref{eq:filter_evaluation} has typically an exponentially growing number of terms, with base $4/3$. 

However, using the fact that for product $C$ the entire filter function factorizes, allows one to calculate it in $O(n)$ time for all $\alpha$.  
For geometrically local Clifford circuits in 1D, using tensor network algorithms we can maintain this efficiency in $n$ with exponential scaling in the circuit depth.  
From \cite{schuster2025random}, we expect that for depth $O(\log(n))$ the approximation with uniformly random global stabilizer states becomes controlled. This suggests that a tensor network algorithm might stay efficient until the regime where direct evaluation becomes admissible in expectation for the Clifford circuits studied here. We leave the formal study of this question to future work.

\section{Fitting procedures}
\label{appendix:fit}
Appropriately choosing the part of the signal where the mixing components are sufficiently suppressed is not a simple task. The mixing length cannot be estimated from noiseless simulations of the $n$-qubit circuits (although it gives an approximation), since the mixing is partially determined by the noise. We confirmed this empirically by evaluating (see Appendix \ref{appendix:signal_sim}) the first and second moment for brickwork circuits of the same form as described in Figure \ref{fig:fig1_explanation}b for different noise models. As an example, we compare the signals for noiseless circuits ($n=20$, Clifford ensemble) to a noisy case, where every brick is followed by single-qubit depolarizing noise on each qubit, see Fig. \ref{fig:mix}. Note in particular in  Fig. \ref{fig:mix}a the significant change in mixing length for $\alpha=1$.\\

\begin{figure}[h]
\refstepcounter{figure} 
\setcounter{mycounterb}{0}
\centering
    \begin{minipage}[b]{5cm}
        \centering
        \includegraphics[width=5cm]{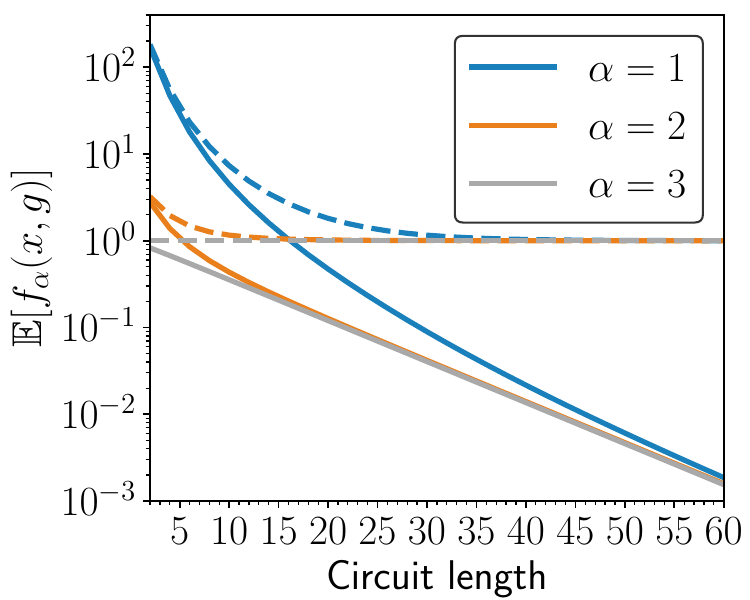}
        \refstepcounter{mycounterb}\label{fig:mixmean}
        \footnotesize(\themycounterb)
    \end{minipage}
    \hspace{0.6cm}
    \begin{minipage}[b]{5cm}
        \centering
        \includegraphics[width=5cm]{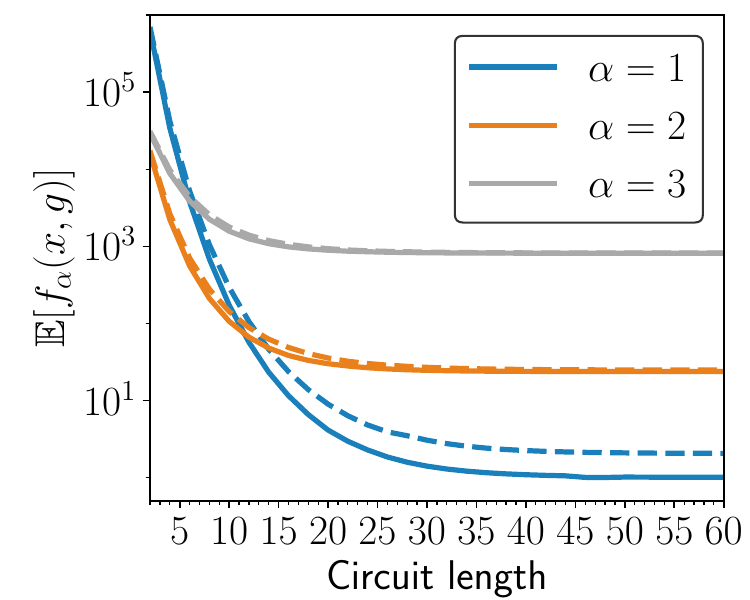}
        \refstepcounter{mycounterb}\label{fig:mixvar}
        \footnotesize(\themycounterb)
    \end{minipage}
    \hspace{0.6cm}
    \begin{minipage}[b]{5cm}
        \centering
        \includegraphics[width=5cm]{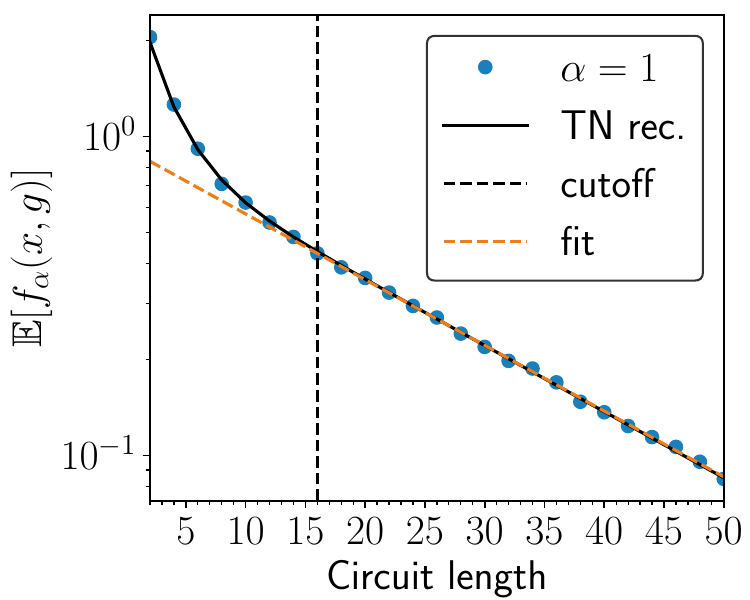}
        \refstepcounter{mycounterb}\label{fig:heuristic}
        \footnotesize(\themycounterb)
    \end{minipage}
\addtocounter{figure}{-1}
\caption{\label{fig:mix} Estimating the mixing length in practice. First (a) and second (b) moments of the RB signal for a brickwork
circuit of the same construction as described in Figure
\ref{fig:fig1_explanation}b with the Clifford ensemble and $n=20$. We show
results for noiseless circuits (dashed line) and circuits where each
brickwork is followed by single qubit depolarizing noise
($p=7.5\cdot 10^{-3}$) on each qubit (continuous line).
(c) The cautious fitting approach applied to the 6-qubit RB data obtained on
the emulator of IBM Guadalupe, the same data as discussed in Figure
\ref{fig:wide2}. The cutoff is estimated by the TN-based reconstruction
method. A single exponential is fitted through the data in the fitting
interval.}
\end{figure}

We considered a variety of heuristic methods for estimating mixing lengths based on experimental data and in this section we discuss three of these protocols. First, we will explain the three methods and then we will show the results we obtain when applying these methods to real RB data.

\subsubsection{\texorpdfstring{Method one: cutoff based on $\chi^2$}{chi2}}
An intuitive statistical approach for determining the start of the fitting interval is to look for the point from which the quality of fit using a single exponential is not significantly worse than a multi-exponential fit. To determine the appropriate starting point for the fitting interval, we perform both a single-exponential and a double-exponential fit for all possible fitting starting points and compare the quality of the fits using a $\chi^2$-metric defined for log-space \begin{equation}\label{eq:chi}
    \chi^2=\sum_i\frac{y_i(\log{y_i}-\log{\hat{y}_i})^2}{\sigma_i^2},
\end{equation}
where $y_i$ is the $i$-th datapoint and $\hat{y}_i$ is the corresponding point predicted by the fit. We define the cutoff as the smallest starting point for which the ratio $r$ of the single-exponential and double-exponential fit obeys $1-\epsilon\leq r \leq 1+\epsilon$, for a small positive cutoff parameter $\epsilon$.
\begin{figure*}[htbp]
    \centering

    \newlength{\plotW}
    \newlength{\plotH}
    \newlength{\gapX}
    \newlength{\gapY}
    \newlength{\labelSpace}
    \newlength{\boxPad}

    \setlength{\plotW}{5cm}
    \setlength{\plotH}{4cm}

    \setlength{\gapX}{0.20cm}

    \setlength{\gapY}{0.30cm}

    \setlength{\labelSpace}{0.40cm}

    \setlength{\boxPad}{0.20cm}

    \newlength{\rowPitch}
    \setlength{\rowPitch}{%
        \dimexpr\plotH+\labelSpace+\gapY\relax
    }

    \begin{tikzpicture}

        \node[
            font=\bfseries,
            anchor=south
        ] at (2.50cm,0.30cm)
        {Cutoff comparison};

        \node[
            font=\bfseries,
            anchor=south
        ] at (7.70cm,0.30cm)
        {Big dataset regime};

        \node[
            font=\bfseries,
            anchor=south
        ] at (12.90cm,0.30cm)
        {Small dataset regime};

        \node[
            anchor=north west,
            inner sep=0pt
        ] (A) at (0,0)
        {%
            \begin{minipage}[t]{5cm}
                \centering
                \includegraphics[
                    width=5cm,
                    height=4cm
                ]{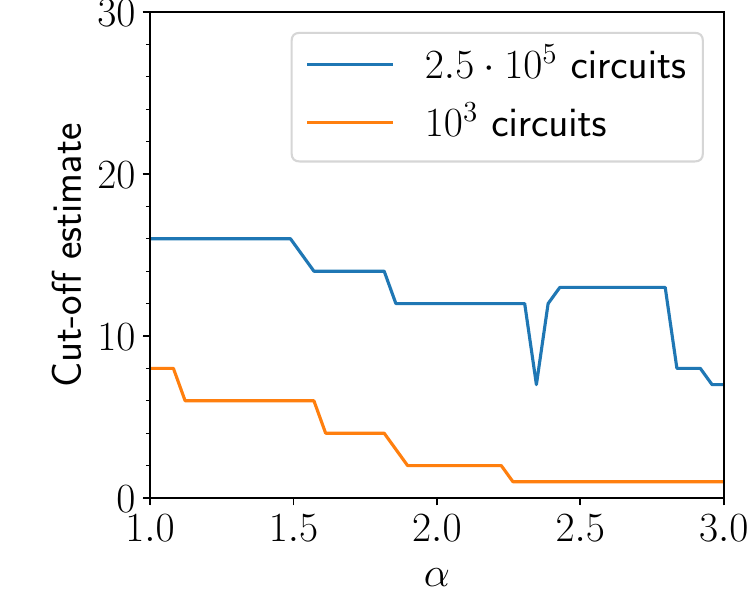}

                \vspace{0.06cm}

                \footnotesize (a)
            \end{minipage}%
        };

        \node[
            anchor=north west,
            inner sep=0pt
        ] (B) at (5.20cm,0)
        {%
            \begin{minipage}[t]{5cm}
                \centering
                \includegraphics[
                    width=5cm,
                    height=4cm
                ]{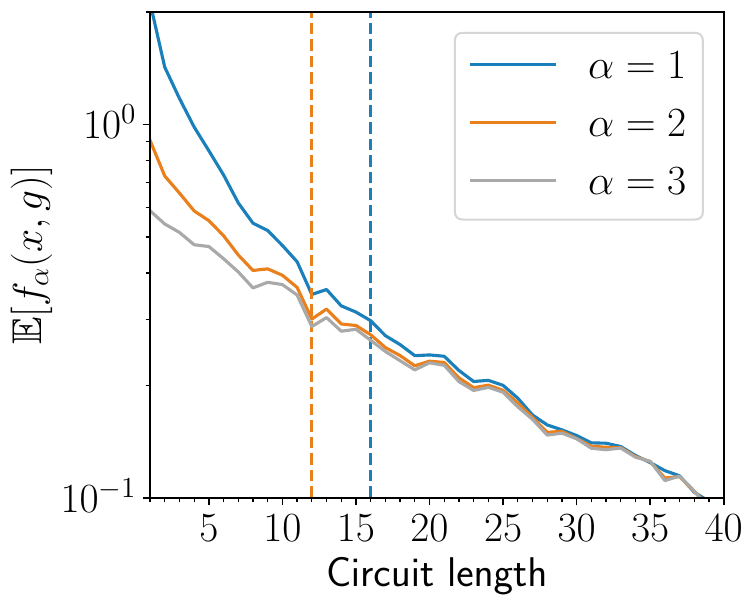}

                \vspace{0.06cm}

                \footnotesize (b)
            \end{minipage}%
        };

        \node[
            anchor=north west,
            inner sep=0pt
        ] (C) at (10.40cm,0)
        {%
            \begin{minipage}[t]{5cm}
                \centering
                \includegraphics[
                    width=5cm,
                    height=4cm
                ]{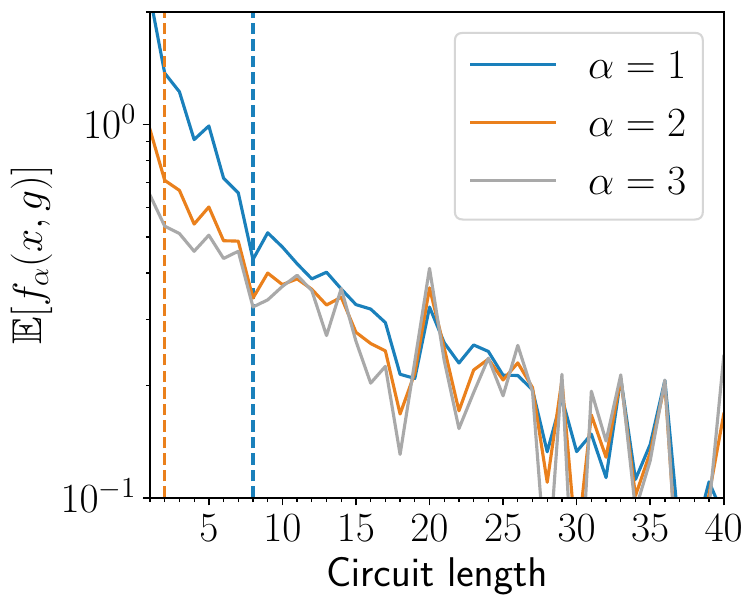}

                \vspace{0.06cm}

                \footnotesize (c)
            \end{minipage}%
        };

        \node[
            anchor=north west,
            inner sep=0pt
        ] (D) at (0,-\rowPitch)
        {%
            \begin{minipage}[t]{5cm}
                \centering
                \includegraphics[
                    width=5cm,
                    height=4cm
                ]{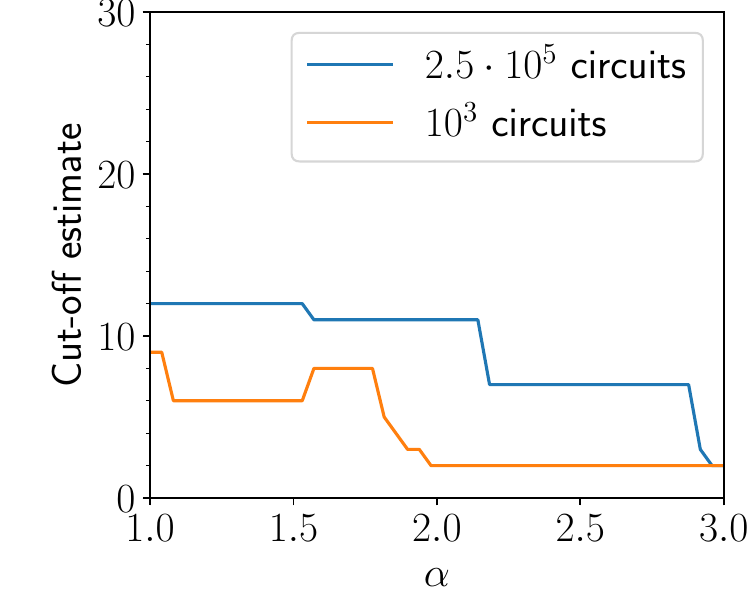}

                \vspace{0.06cm}

                \footnotesize (d)
            \end{minipage}%
        };

        \node[
            anchor=north west,
            inner sep=0pt
        ] (E) at (5.20cm,-\rowPitch)
        {%
            \begin{minipage}[t]{5cm}
                \centering
                \includegraphics[
                    width=5cm,
                    height=4cm
                ]{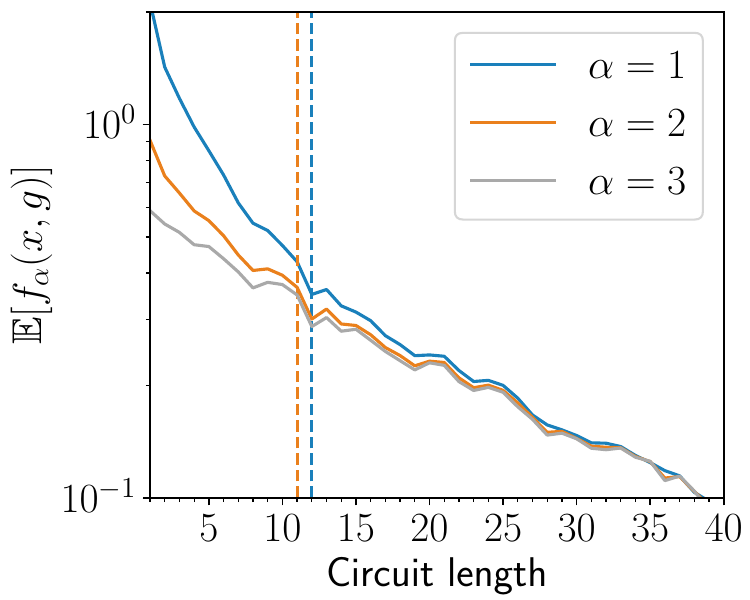}

                \vspace{0.06cm}

                \footnotesize (e)
            \end{minipage}%
        };

        \node[
            anchor=north west,
            inner sep=0pt
        ] (F) at (10.40cm,-\rowPitch)
        {%
            \begin{minipage}[t]{5cm}
                \centering
                \includegraphics[
                    width=5cm,
                    height=4cm
                ]{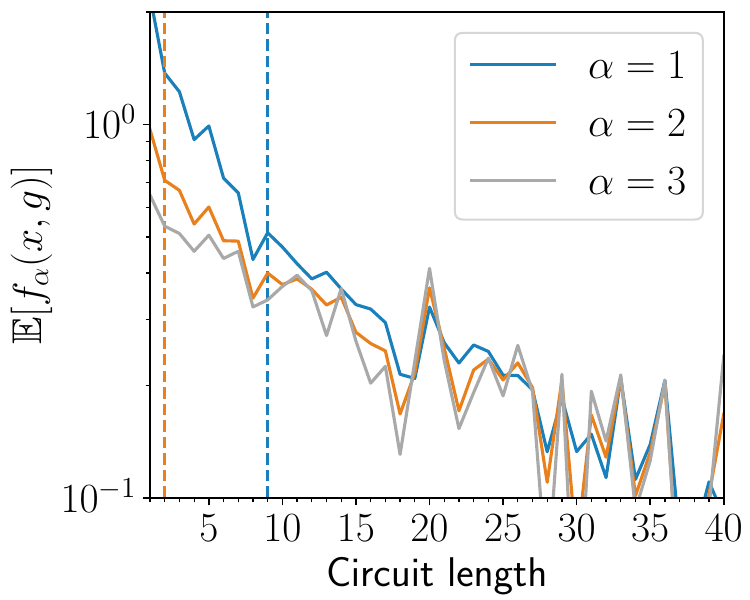}

                \vspace{0.06cm}

                \footnotesize (f)
            \end{minipage}%
        };

        \node[
            anchor=north west,
            inner sep=0pt
        ] (G) at (0,-2\rowPitch)
        {%
            \begin{minipage}[t]{5cm}
                \centering
                \includegraphics[
                    width=5cm,
                    height=4cm
                ]{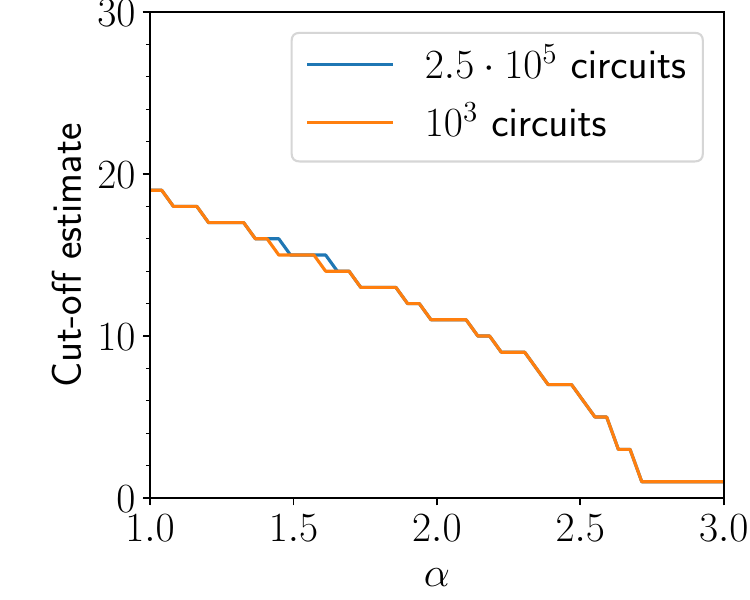}

                \vspace{0.06cm}

                \footnotesize (g)
            \end{minipage}%
        };

        \node[
            anchor=north west,
            inner sep=0pt
        ] (H) at (5.20cm,-2\rowPitch)
        {%
            \begin{minipage}[t]{5cm}
                \centering
                \includegraphics[
                    width=5cm,
                    height=4cm
                ]{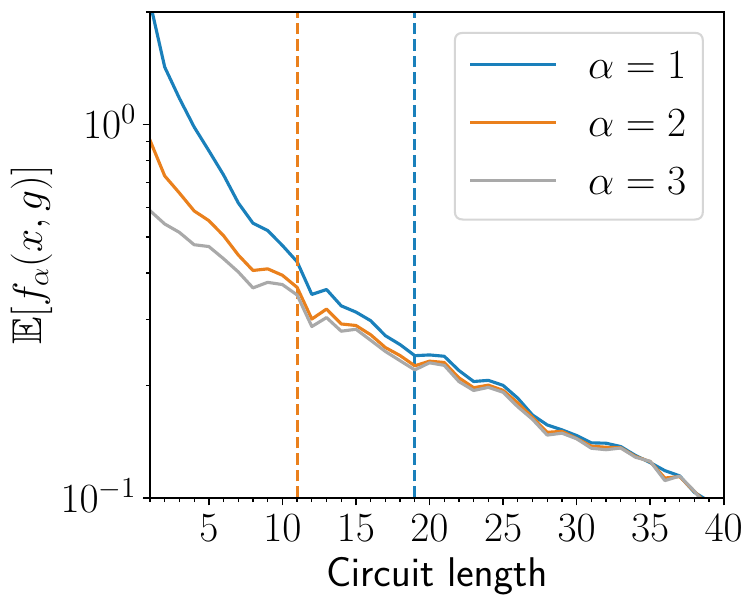}

                \vspace{0.06cm}

                \footnotesize (h)
            \end{minipage}%
        };

        \node[
            anchor=north west,
            inner sep=0pt
        ] (I) at (10.40cm,-2\rowPitch)
        {%
            \begin{minipage}[t]{5cm}
                \centering
                \includegraphics[
                    width=5cm,
                    height=4cm
                ]{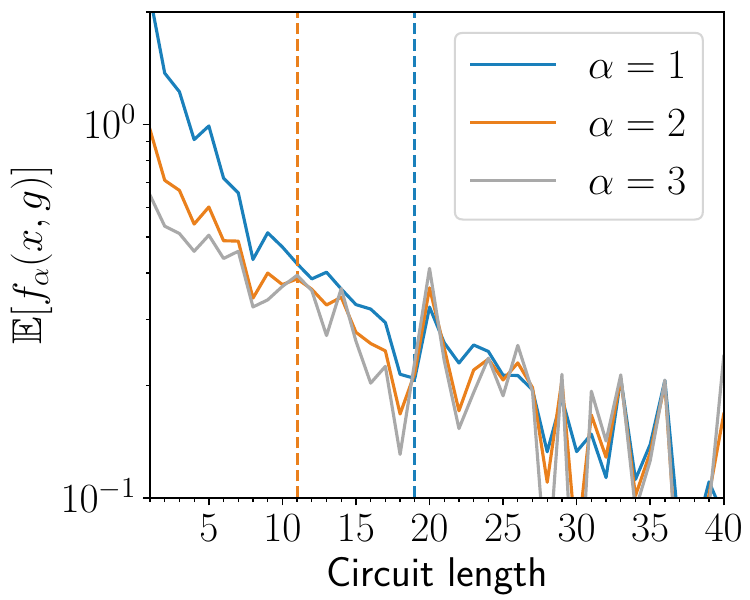}

                \vspace{0.06cm}

                \footnotesize (i)
            \end{minipage}%
        };

        \node[
            font=\bfseries,
            anchor=center
        ] at (-0.60cm,-2.20cm)
        {\rotatebox{90}{$\mathbf{\chi^2}$}};

        \node[
            font=\bfseries,
            anchor=center
        ] at (-0.60cm,{-\rowPitch-2.20cm})
        {\rotatebox{90}{SVD Hankel matrix}};

        \node[
            font=\bfseries,
            anchor=center
        ] at (-0.60cm,{-2\rowPitch-2.20cm})
        {\rotatebox{90}{TN reconstruction}};

        \draw[
            rounded corners=4mm,
            line width=0.9pt
        ]
        (-0.20cm,0.20cm)
        rectangle
        (15.60cm,
         {-3\rowPitch+\gapY-0.20cm});

        \draw[
            line width=0.55pt
        ]
        (5.10cm,0.20cm)
        --
        (5.10cm,{-3\rowPitch+\gapY-0.20cm});

        \draw[
            line width=0.55pt
        ]
        (10.30cm,0.20cm)
        --
        (10.30cm,{-3\rowPitch+\gapY-0.20cm});

        \draw[
            line width=0.55pt
        ]
        (-0.20cm,{-\plotH-\labelSpace-\gapY/2})
        --
        (15.60cm,{-\plotH-\labelSpace-\gapY/2});

        \draw[
            line width=0.55pt
        ]
        (-0.20cm,{-\rowPitch-\plotH-\labelSpace-\gapY/2})
        --
        (15.60cm,{-\rowPitch-\plotH-\labelSpace-\gapY/2});

    \end{tikzpicture}

    \caption{\label{fig:method_comparison}
        Comparison of different methods to determine the cutoff point. For each method, we show the results when applied to two different RB curves. The RB curves originate from the 6-qubit Rigetti experiment with the icosahedral ensemble described in the main text. We use the results from both an experiment with $2.5 \cdot 10^5$ circuits per sequence length and from an experiment with $10^3$ circuits per sequence length to investigate the robustness of the different methods. In panels (a,d,g), we show how the cutoff points depend on the filter function. In panels (b,e,h) we visualize the cutoff points for $\alpha=1$ and $\alpha=2$ for the experiment with $2.5 \cdot 10^5$ circuits and we do the same in panels (c,f,i) for the experiment with $10^3$ circuits per sequence length.
    }

\end{figure*}

\subsubsection*{Method two: cutoff based on SVD of the Hankel matrix}
Matrix pencil methods form a specialized group of algorithms that aim to approximate a function by a sum of complex exponentials \cite{potts2013parameter,sarkar1995using}. These methods rely on the Hankel matrix of the signal and it is well known that the singular values of this Hankel matrix describe the relative importance of the different exponential decays that are reconstructed (see Ref. \cite{sarkar1995using}). Typically, there will be some dominant singular values and a large group of non-dominant singular values that are all of similar magnitude, indicating that these originate from noise in the signal. Therefore, a natural way to determine the cutoff point such that the remaining signal can be well described by a single-exponential decay is to find the point after which there is only a single dominant singular value. We define the start of the fitting interval as the first point for which the ordered singular values $\sigma_1,\sigma_2,\hdots$ of the Hankel matrix obey $ \frac{\log{\sigma_2}}{\log{\sigma_3}}\geq 1-\delta$, for a small positive cutoff parameter $\delta$. 
\subsubsection*{Method three: cutoff based on tensor network reconstruction}
For specific types of circuits and noise models, we can simulate the RB signal using tensor networks (see Appendix \ref{appendix:signal_sim}). Given such a simulated RB curve, it is easy to determine the cutoff point, since this curve can be calculated exactly and can be extended to sufficient length using the simulations. We propose to use a search algorithm to find a noise model for which the $\chi^2$ in log-space (see Equation \ref{eq:chi}) between the simulated curve and the experimental data is minimal. We restrict the noise model to depolarizing noise, which means that we can use a simple algorithm, based on the curve-fitting package in scipy, to minimize this metric and find the best fit. This approximate noise model is then used to simulate the RB signal, from which we can easily and accurately extract the cutoff points by comparing the singular values of the Hankel matrix (see previous section). An application of this method to data from an emulation of the IBM Guadalupe device is shown in Figure \ref{fig:heuristic}. 

\subsubsection{Comparison}
We compare the described methods for determining the cutoff in Figure \ref{fig:method_comparison}. We apply the methods to data from the 6-qubit Rigetti experiment using the icosahedral ensemble described in the text. Ideally, the methods should be able to find a reasonable cutoff even in a low signal-to-noise regime, which is why we show results for both the full dataset ($2.5 \cdot 10^5$ circuits per sequence length) and a subset of the data ($10^3$ circuits per sequence length). For each method we show the estimated cutoff points for both datasets and we show examples of cutoff points overlaid with the RB signals. We notice that the $\chi^2$ and singular value method break down when provided with a noisier signal and have a tendency to underestimate the mixing time. These methods are pure statistical methods (with the singular value method being more specialized towards exponential signals and the $\chi^2$ method being more generally applicable) and can not be reliably used when there is not enough data. When solely provided with the plots in Figure \ref{fig:method_comparison}b and Figure \ref{fig:method_comparison}e, the cutoffs might seem reasonable, but the plots in Figure \ref{fig:method_comparison}c and Figure \ref{fig:method_comparison}f show that this is not the case. Hence substantial care must be taken when using these methods in a situation with small amounts of available data.

The TN-based method is heuristic, but incorporates more knowledge about the mixing behavior of the circuits used in the experiment and the estimated cutoffs are observed to be more resistant to noisy data. Caution is still required, but we consider the tensor network based method to be the best and most stable option for the experiments that we performed.

\section{Tetrahedral ensemble}
In this appendix we show the results for the experiment performed on the Rigetti device using the tetrahedral ensemble. 
\begin{figure*}[h]
\refstepcounter{figure} 
\setcounter{mycounterb}{0}
        \begin{minipage}{0.9\linewidth}
   \begin{minipage}[b]{5cm}
            \centering
            \includegraphics[width=1\textwidth]{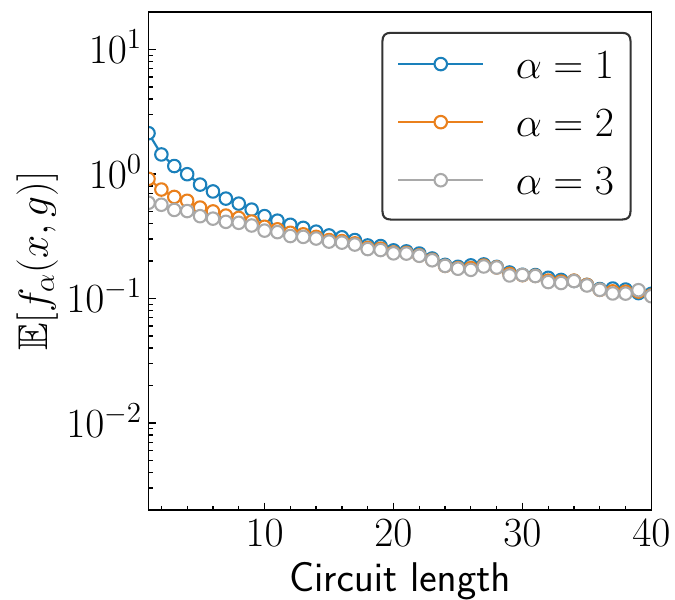}
    \refstepcounter{mycounterb}
    \hspace*{0.75cm}\footnotesize(\themycounterb)
        \end{minipage}
        \hfill
        \begin{minipage}[b]{5cm}  
            \centering 
            \includegraphics[width=1\textwidth]{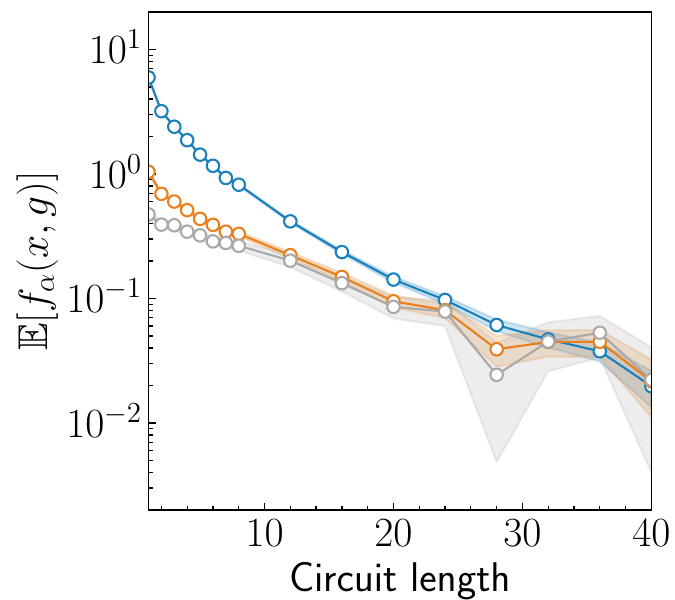}
        \refstepcounter{mycounterb}
    \hspace*{0.75cm}\footnotesize(\themycounterb)
        \end{minipage}
        \hfill
        \begin{minipage}[b]{5cm} 
            \centering 
            \includegraphics[width=1\textwidth]{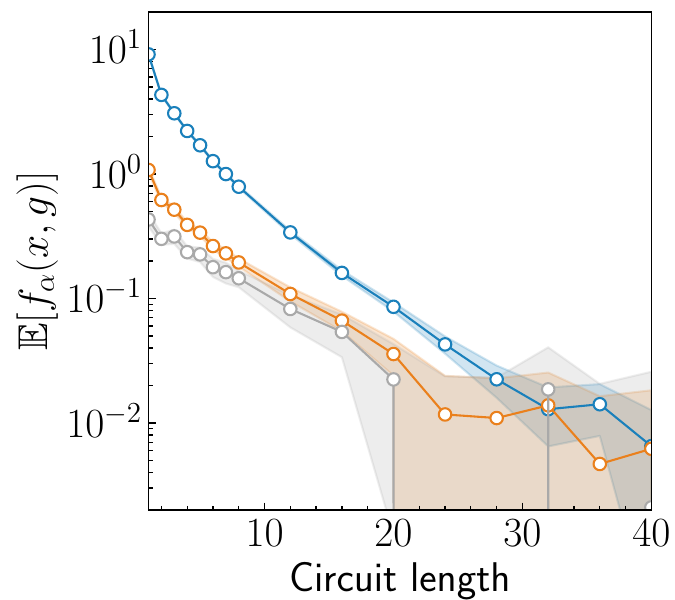}
    \refstepcounter{mycounterb}
    \hspace*{0.75cm}\footnotesize(\themycounterb)
        \end{minipage}
    \end{minipage}
\addtocounter{figure}{-1}
\caption{\label{fig:tetra} Randomized benchmarking experiments performed on the Cepheus-1-36Q Rigetti device. We use random circuits as described in Figure \ref{fig:fig1_explanation}b and show results for the tetrahedral ensemble. Per sequence length $m$, $2.5 \cdot 10^5$ circuits are used (with a single shot per circuit) for the $n=6$ case and $10^5$ circuits for the $n=10/12$ case. For $n=6$, $m \in \{1\hdots 40\}$ and for $n=10/12$ we have $m \in \{1,  2,  3,  4,  5,  6,  7,  8, 12, 16, 20, 24, 28, 32, 36, 40\}$. Uncertainties are estimated using bootstrapping procedures and shaded areas represent 95\%-uncertainty intervals.}
\end{figure*}
\section{IBM Guadalupe simulation}
In this appendix we display for completeness the simulation data generated from the Qiskit emulator of the $6$--qubit Guadalupe device, which was used to generate the relative uncertainty plot vis-a-vis the parameter $\alpha$ in Fig. \ref{fig:figure2}d in the main text. The simulated data tell broadly the same story as the $n=30$ numerical simulations displayed in Fig~\ref{fig:figure2} in the main text and our experimental results discussed in Figure \ref{fig:fig3_rigetti_exp}.

\begin{figure*}[htbp]
\refstepcounter{figure} 
\setcounter{mycounterb}{0}
        \begin{minipage}{0.9\linewidth}
   \begin{minipage}[b]{5cm}
            \centering
            \includegraphics[width=1\textwidth]{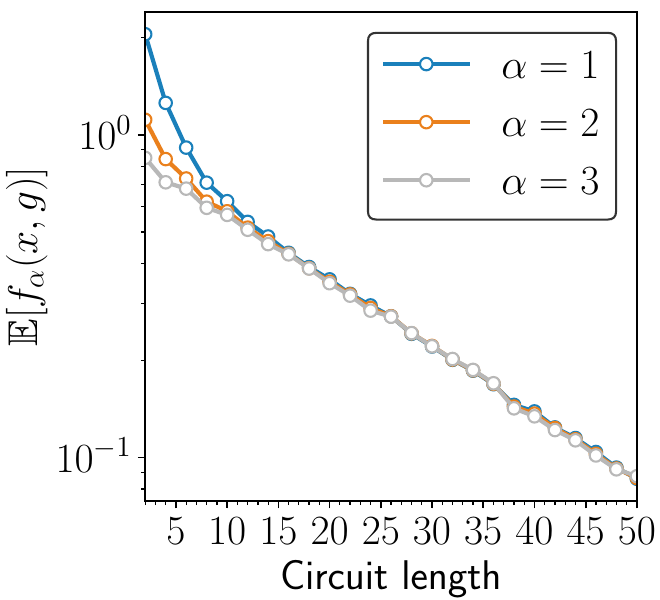}
    \refstepcounter{mycounterb}\label{fig:RBQISKmean}
    \hspace*{0.75cm}\footnotesize(\themycounterb)
        \end{minipage}
        \hfill
        \begin{minipage}[b]{5cm}  
            \centering 
            \includegraphics[width=1\textwidth]{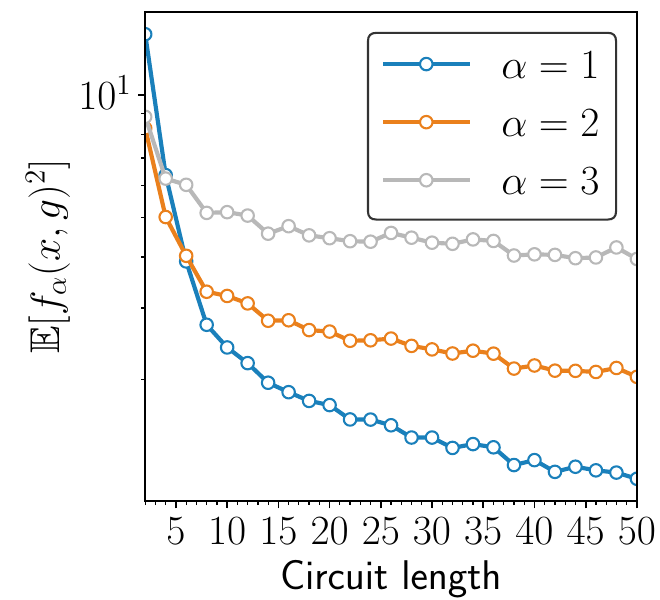}
        \refstepcounter{mycounterb}\label{fig:RBQISKdec}
    \hspace*{0.75cm}\footnotesize(\themycounterb)
        \end{minipage}
        \hfill
        \begin{minipage}[b]{5cm} 
            \centering 
            \includegraphics[width=1\textwidth]{Decay_vs_n6.pdf}
    \refstepcounter{mycounterb}\label{fig:RBQISKvar}
    \hspace*{0.75cm}\footnotesize(\themycounterb)
        \end{minipage}
    \end{minipage}
\addtocounter{figure}{-1}
\caption{\label{fig:wide2} Results for generator RB using brickwork circuits of local 2--qubit Cliffords for a 6--qubit system on the emulator of IBM Guadalupe. For each sequence length, $1.2 \cdot 10^5$ circuits are sampled, with 100 shots per circuit.
Panel (a) shows the signal for a selection of filter functions and panel (b) shows the second moments. The decay values extracted using the two different fitting strategies are shown in panel (c).}
\end{figure*}

\end{document}